\documentclass{article} 
\usepackage{iclr2027_conference,times}
\usepackage{hyperref}
\hypersetup{pdfborder={0 0 0}}
\renewcommand{\iclrruler}[1]{}
\usepackage{url}
\usepackage{amsmath,amssymb,amsthm}
\usepackage{booktabs,multirow,tabularx}
\usepackage{graphicx}
\usepackage{stfloats}
\usepackage{float}
\usepackage{xspace}
\usepackage{xcolor}
\newcommand{\upd}[1]{\textcolor{cyan}{#1}}
\newcommand{\obs}[1]{\textcolor{orange}{#1}}

\newtheorem{theorem}{Theorem}[section]
\newtheorem{proposition}[theorem]{Proposition}

\newcommand{\Kn}{\mathrm{Kn}}
\newenvironment{AIdeclaration}{\paragraph*{AI use declaration}}{}
\newenvironment{Reprostatement}{\paragraph*{Reproducibility statement}}{}

\let\cite\citep

\makeatletter
\g@addto@macro\normalsize{%
  \setlength{\abovedisplayskip}{5pt plus 1pt minus 2pt}%
  \setlength{\belowdisplayskip}{5pt plus 1pt minus 2pt}%
  \setlength{\abovedisplayshortskip}{3pt plus 1pt minus 1pt}%
  \setlength{\belowdisplayshortskip}{3pt plus 1pt minus 1pt}%
}
\makeatother
\title{A Kinetic Theory of the Gated Self-Evolving LLM Agent}
\author{Haipeng Wang\\
Neusoft Corporation\\
\texttt{haipeng\_wang@neusoft.com}}

\begin{document}

\maketitle

\begin{abstract}

We find traces of fluid dynamics in the self-evolution of an LLM agent, and give the kinetic theory that predicts them. Gated self-evolution is the loop in which an agent rewrites its own skills under a validation gate. Self-evolution research has treated the agent as the unit; we study instead the individual instances inside it. Here the agent is DSH-plugin-based: it runs in production on DeepSeek Harness (DSH), and its plugins satisfy four architectural properties (permutation symmetry, reversibility, acyclicity, typed contracts), which license treating these instances as identical hard spheres; the theory is accordingly scoped to DSH-class plugin populations. On this scope the paper builds three theory layers. The rigorous layer, independent of any analogy, comprises an any-time hitting-time certificate bounding the expected rounds to any prescribed improvement, a resolution law that prices held-out validation budgets, and a separation theorem: the daemon must stay outside the population, because merging evaluator with evaluated voids the certificate. The kinetic layer is a master equation over the plugin~$\times$~version~$\times$~task grid with four operators (collision, reaction, external field, gate), where collision is co-activation. Its moment hierarchy, the step that turns a gas into fluid equations, generates the falsifiable statistical signatures. Throughout, the fluid reading is a bounded analogy: momentum is \emph{not} conserved, so no Navier--Stokes limit exists. The measured layer runs on a faithful minimal instance, a large library of four-parameter skill plugins retrieved one per episode with a co-activation probe, in a one-model, one-task-family WebShop environment; every element maps to the DSH loop by architectural role. Population fluctuation scaling is density-gated: invisible at sparse edit density, it emerges at the predicted rate under tripled density, as directional evidence.

\end{abstract}

\section{Introduction}

Self-evolving LLM agents write their own skills and tools from trajectory evidence, mount them into a registry, and keep a change only when it strictly improves a held-out validation score~\cite{luo2026skillopt}. Prior work on such agents has treated the agent itself as the unit of study~\cite{shinn2024reflexion,yin2025godel,sgm2025statistical}. We looked instead at what happens inside the registry during the agent's run, where a population of individual instances is repeatedly retrieved into episodes, edited, and accepted or rejected by the gate. When this population is large, the running of these instances carries traces of fluid dynamics in its statistics, much like a gas moving under an external selector. The natural language for such behavior is the kinetic theory of gases, today grounded in the derivation of Boltzmann's equation from hard-sphere dynamics~\cite{lanford1975,deng2024longtime}.

We develop this observation into a kinetic theory of the gated plugin population. In studying \textsc{DSH}~\cite{deepseekharness}, an open-source ``everything is a plugin'' agent harness, we found that a production two-tier system runs on it. Its Cordis kernel~\cite{shi2026composability} proves four structural properties of its plugins, and its architecture has exactly these four properties: permutation-symmetric composition, mount/unmount reversibility, acyclic dependencies, and typed contracts. It was under exactly these properties that we felt the plugin-version instances behave as identical hard spheres. A monolithic self-modifying agent satisfies none of them, and for it the kinetic description does not apply. What distinguishes this gas from the nearest-neighbor populations is the gate. It acts as a Maxwell daemon external to the gas, turning diffusion into optimization and yielding both a convergence certificate and decisive statistical signatures.

The claim is then closed from both sides. On the theory side, the statements are machine-checked theorems: a hitting-time certificate that bounds the expected rounds to a prescribed improvement, a resolution law that prices the held-out validation budget, and a separation theorem that keeps the evaluator outside the population. On the measurement side, production-scale training is too slow for the statistics needed, so all evidence comes from a \emph{faithful minimal instance} of the same loop (dynamical equivalence not claimed), whose pre-registered probes return directional evidence for the predicted traces. The gas-like statistics are real but density-conditional, invisible at sparse edit density and emerging at the predicted rate only under tripled density. Throughout, the fluid picture is read as a bounded analogy. The arc thus closes where it opened: the self-evolution of a single LLM agent, hosting a large population of instances, carries traces of fluid dynamics in its statistics.

\paragraph{Contributions.}
\begin{itemize}
  \item We find traces of fluid dynamics in the self-evolution of an LLM agent, where a large plugin population under a validation gate exhibits density-gated fluctuation scaling and relaxation to a non-equilibrium steady state, as directional evidence.
  \item We give a kinetic theory that predicts these signatures, a master equation over the plugin~$\times$~version~$\times$~task grid in which collision is co-activation and the validation gate acts external to the gas, with a moment hierarchy that generates the falsifiable statistical predictions.
  \item We prove, under the four architectural properties of DSH, a hitting-time certificate for gated self-evolution, a resolution law that prices held-out validation budgets, and a separation theorem that keeps the evaluator outside the population.
\end{itemize}

\section{Related Work}\label{sec:related}

\textbf{Kinetic and physics-flavored theories.} The most closely related work is that of Jung et al.~\cite{jung2026}, whose kinetic theory of collective learning, a Boltzmann-type equation for an adaptive-agent population, shares the population-level viewpoint. We differ in three ways. Our gate supplies external selection, without which drift decouples from the objective; our memory is a shared retrieval pool acting as deliberate re-collision rather than a single-particle filter; and our readings anchor to a production runtime and a measured device rather than to simulation. Mean-field MARL~\cite{mfmar2018} uses the mean field as a computational approximation; in our theory the same closure is the statistical level of description for gated, non-optimizing skill instances. Pareschi--Toscani~\cite{pareschi1999} supplies the mathematical framework for this closure. De Marzo et al.~\cite{demarzo2026majority} show gate-free majority-following laws in LLM agent societies, complementary to our driven-dissipative kinetics. Model collapse~\cite{shumailov2024collapse} is the no-gate limit; physics-flavored accounts analogize single-model optimization, none deriving long-time machinery from a population-level gate structure. Throughout, classical phenomena serve as hypothesis generators and interpretations are labeled.

\textbf{Self-evolving agents.} SkillOpt~\cite{luo2026skillopt} is the discipline we inherit: skill document as trainable state, bounded edits, strict gate, rejected-edit buffer; it defines the \emph{training loop} whose population statistics we derive, and its strict gate is the selection operator our theory formalizes. RE-TRAC~\cite{retrac2026} (memory conditioning) and Resource2Skill~\cite{resource2skill2026} (skill distillation) are orthogonal. Reflexion~\cite{shinn2024reflexion} and the G\"odel Agent~\cite{yin2025godel} update task/meta policies by self-critique without external statistical control; our certificate is what that line lacks. The companion two-tier runtime (33 plugins, built on DeepSeek Harness~\cite{deepseekharness}) is the concrete embodiment.

\textbf{Statistical gating.} SGM~\cite{sgm2025statistical} recasts recursive self-modification as gambling with e-wealth supermartingales; SEA~\cite{sengupta2026sea} carries e-values to LLM self-evolution. Both stand on the anytime-valid statistics lineage of sequential tests~\cite{wald1945}, time-uniform confidence sequences~\cite{howard2021time}, and game-theoretic e-processes~\cite{ramdas2023game}. That line answers when to stop safely, given paired differences. Two open questions are central here: (i) \emph{power}, namely the resolution limit $\Delta_{\min}(n,\hat\sigma)$ of a finite held-out set, which is what actually binds late-training improvement (our Proposition~D companion); and (ii) \emph{multi-candidate selection bias}, since planners proposing 5--8 pipelines per round break the single-candidate assumption, and the Vovk--Wang merge~\cite{vovk2021e} must be invoked (made explicit in our gate). We claim no priority in the statistics of gating; our claim is the gate's role \emph{as an external selection operator} in a population-level kinetic theory, which transforms gate design from a testing procedure into a structural axiom.

\section{Overview}\label{sec:overview}

\begin{figure}[t]
\centering
\includegraphics[width=0.85\textwidth]{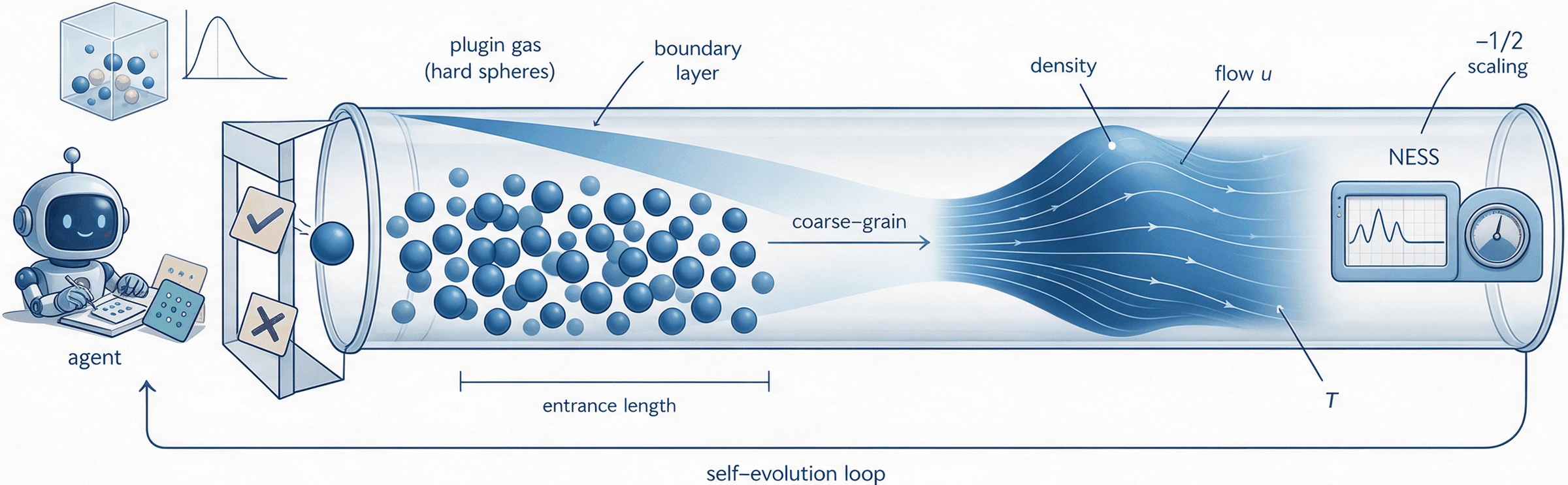}
\caption{Graphical overview of the proposed kinetic theory: the self-evolution run of a gated LLM agent, interpreted through the framework of fluid mechanics.}

\label{fig:overview}
\end{figure}

We develop a statistical theory for the self-evolution of a single LLM agent. The agent's registry maintains a large population of plugin-version instances under a validation gate, and we show that this population behaves, measurably, like a gas worked on by an external selector; the kinetic theory of such a gas predicts what the loop's statistics show. Figure~\ref{fig:overview} depicts this population. The \emph{agent} edits plugin cards; the validation gate, acting as a Maxwell daemon outside the gas, admits improving cards into the \emph{plugin gas} of hard spheres; past the \emph{coarse-grain} transition the gas reads as a continuum with density $\rho(x,v,e,t)$, flow $u$, and temperature $T$; the instruments at the right end measure the predicted traces, $-1/2$ fluctuation scaling and \emph{NESS} relaxation, as directional evidence. Every fluid element is axiom-licensed, and the fluid picture throughout is a bounded analogy.

\section{Theory}\label{sec:theory}

\subsection{State space and the kinetic equation}\label{sec:ktp-eq}

The particle picture of Figure~\ref{fig:overview} is grounded in architecture. It came out of our study of the DSH plugin system~\cite{deepseekharness}: its Cordis kernel~\cite{shi2026composability} proves four structural properties of the plugins. It struck us that under exactly these properties the plugin-version instances behave as identical hard spheres. The four properties are theorems of the kernel, not developer conventions, and we adopt them as the axioms A1--A4 of the kinetic description (the precise Cordis Def./Thm. correspondences are tabulated in Appendix~\ref{app:labels}):

\begin{itemize}\itemsep1pt
\item[A1] \emph{Permutation symmetry:} plugin composition is order-independent, so identical particles are interchangeable.
\item[A2] \emph{Mount/unmount reversibility:} plugins attach and detach without residue, the microscopic reversibility of the population.
\item[A3] \emph{Acyclic dependencies:} no instance instantaneously influences itself; interactions follow a single-time-slice causal order.
\item[A4] \emph{Typed contracts:} interactions are bounded and typed, the short-range contact of the gas.
\end{itemize}

A monolithic self-modifying agent satisfies none of A1--A4, and for it the kinetic description is undefined, not merely inaccurate. The everyday alternatives fail as well. A shared skill document violates A1--A3 and offers no interface types; an MCP stack keeps A4 but carries untracked side effects and cyclical calls. Neither provides a versioned instance multiset with carried credit state, so there are no particles to count; the ledger and the formal obstruction are given in Table~\ref{tab:why-dsh} and Proposition~\ref{prop:obstruct}. Given A1--A4, a plugin population is a multiset of version instances $(\xi, v, e)$, with plugin identity $\xi$, parameter vector $v$, and credit record $e$, living on the discrete grid of plugin $\times$ version $\times$ task. We construct the \emph{KTP master equation}: the density $f(x,v,e,t)$ over capability space $x$, parameter momentum $v$, credit $e$, and time $t$ obeys
\begin{equation}\label{eq:ktp}
\partial_t f + \upd{v}\cdot\nabla_x f + \nabla_{\upd{v}}\cdot(F\, f) \;=\; Q[f,f] + R[f] + G[f],
\end{equation}
It is a five-term balance. The flow term is execution transport, deterministic and bit-reproducible (Fig.~\ref{fig:hardsphere}). The external field $F$ is the task distribution pulling the population, with the LLM inference chain acting as its thermostat. The collision term $Q$ is co-activation with attribution, the reaction term $R$ is mount/unmount/materialize, and the gate term $G$ is selective dissipation. A full term-by-term description with production-code locations appears in Appendix~\ref{app:ops}.

\begin{figure}[!t]
\centering
\includegraphics[width=0.62\textwidth]{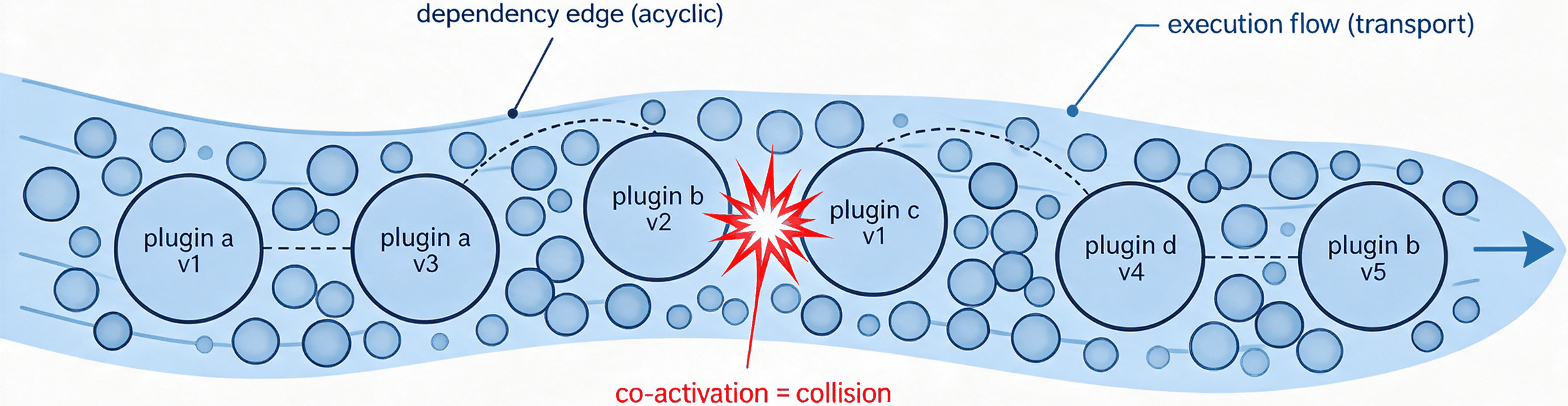}
\caption{Plugin particles carried by the execution flow: transport is reversible execution ordering, not physical inertia; co-activation in one pipeline is the collision, marked in red; dashed edges are the acyclic dependency skeleton; the collision rule is sampled by the LLM reflector.}
\label{fig:hardsphere}
\end{figure}

\subsection{The four operators and the effective population}\label{sec:operators}

\begin{figure}[H]
\centering
\includegraphics[width=0.55\textwidth]{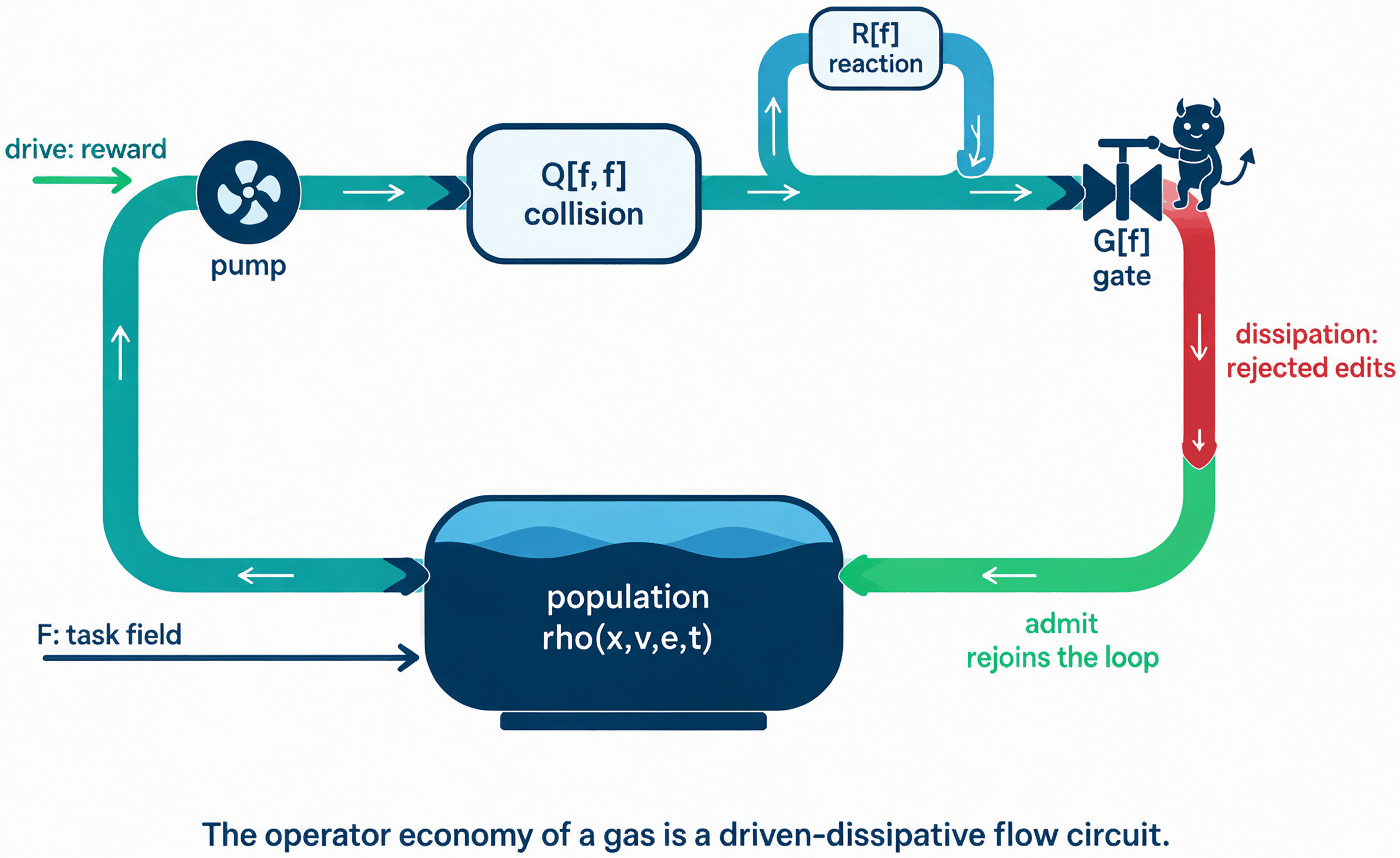}
\caption{The driven-dissipative circuit of the operators. Reward drives the pump into the collision chamber $Q[f,f]$ (fed by the shared experience database), the reaction loop $R[f]$ mutates species, the gate valve $G[f]$ admits improving candidates and drains rejected edits, and the population tank $\rho(x,v,e,t)$ rejoins the loop; $F$ is the external task field. Drive--dissipation balance at the valve is the NESS clause (Section~\ref{sec:ness}).}
\label{fig:circuit}
\end{figure}

The four rules are easiest to read as one particle's life history: it meets others, passes a survival test, lives among changing species, and feels the pull of the task field.
\emph{Collision $Q$}: instances $i,j$ meet when co-activated in a pipeline, with rate kernel $K(x_i,x_j;\tau)$ and cross-section $\sigma_i = \int K(x_i,x';\tau)\,dx'\,d\tau$; the cross-section is bounded structurally by retrieval locality, so an instance collides only within its tag neighborhood. A collision returns batch reward $R_{\mathrm{batch}}$, and the reflector's attribution gives the post-state
\begin{equation}
\upd{e_i'} + \upd{e_j'} = e_i + e_j + R_{\mathrm{batch}},\qquad
\upd{v_i'} = v_i + \eta\, a_i,\qquad
\upd{s_i'} = s_i + \Delta s_i,
\end{equation}
an inelastic, driven-gas event in which the reward is external work and the LLM is the sampler of the collision rule.

\emph{Gate $G$}: every candidate post-state $z'$ is retained iff
\begin{equation}
g(z') = \mathbb{1}\big\{ \obs{V(z')} > \obs{V(z)} \big\},
\end{equation}
with $V$ the held-out score. $G$ is the \emph{only term that projects downhill motion}: all directionality of optimization comes from this projection, and its measurement information decides which particles lose energy; this is the Maxwell daemon outside the gas (Theorem~\ref{thm:demon}; Fig.~\ref{fig:gate}).

\emph{Reaction $R$} (append/mount/unmount/materialize) changes species counts; sandbox trial-mounting is the gas phase and materialization is solidification (Appendix~\ref{app:ops}).

\emph{External field $F$}: the task distribution pulls the population across tasks, and the LLM inference chain acts as its thermostat. The field does no bookkeeping of its own, entering the master equation only through the flux $\nabla_{\upd{v}}\cdot(F\, f)$. Together with $Q$, $R$, and $G$, it closes the driven-dissipative loop (Figure~\ref{fig:circuit}): reward drives collisions, reactions mutate species, the gate dissipates rejected candidates, and admitted post-states rejoin the population; drive--dissipation balance at the valve is the NESS clause (Section~\ref{sec:ness}), and a worked round with recorded numbers appears in Appendix~\ref{app:records}.

The theorems that follow also need the effective particle count $N_{\mathrm{eff}}$, frozen before fitting by a pre-registered gauge (the occupied grid cells, with gauge agreement checked rather than assumed; Appendix~\ref{app:ops}). The rarefaction condition is that no instance dominates, $\max_i \sigma_i / \sum_j \sigma_j \ll 1$; this is the bounded-kernel premise of Theorem~\ref{thm:meanfield}, provided structurally by retrieval locality and probed on the device (a finite-run diagnostic of the premise, not a proof).

\subsection{Dynamics and statistics of the gated gas}\label{sec:results}
\begin{theorem}\label{thm:hitting}
Assume (H1$'$) the reflect kernel proposes a strictly improving edit with probability $\ge p_{\min}$ and gain $\ge \Delta$ whenever such an edit exists; (H2) the gate is a strict held-out inequality on an independent validation set; (H3$'$) the deployed skill is always the latest accepted version and candidates are physically quarantined. Then for the first bounded-improvement threshold $\varepsilon$
\begin{equation}
\mathbb{E}[T_\varepsilon] \;\le\; \frac{V_0}{p_{\min}\,\Delta},
\end{equation}
where $V_0$ is the initial validation gap (score cap minus deployed score), and the best-observed score is monotone along the accepted lineage; the bound is robust to a change of state space (plugin granularity).
\end{theorem}
The proof strategy is direct: along the accepted lineage, the best score is a nonnegative quantity that increases by at least $p_{\min}\Delta$ per accepted round on average, and a standard drift argument turns this per-round drift into the expected hitting bound (full derivation: Appendix~\ref{app:lean}). Operationally, $p_{\min}\Delta$ has a measurable counterpart in the run log---the acceptance rate times the mean accepted gain---so the certificate is read post hoc as a consistency check on the realized drift. By analogy, the bound plays the role of a relaxation timescale of the population flow: it explains the measured saturation---late in training the binding constraint is the resolution floor, not the skill (Section~\ref{sec:experiments}).

\begin{figure}[!t]
\centering
\includegraphics[width=0.42\textwidth]{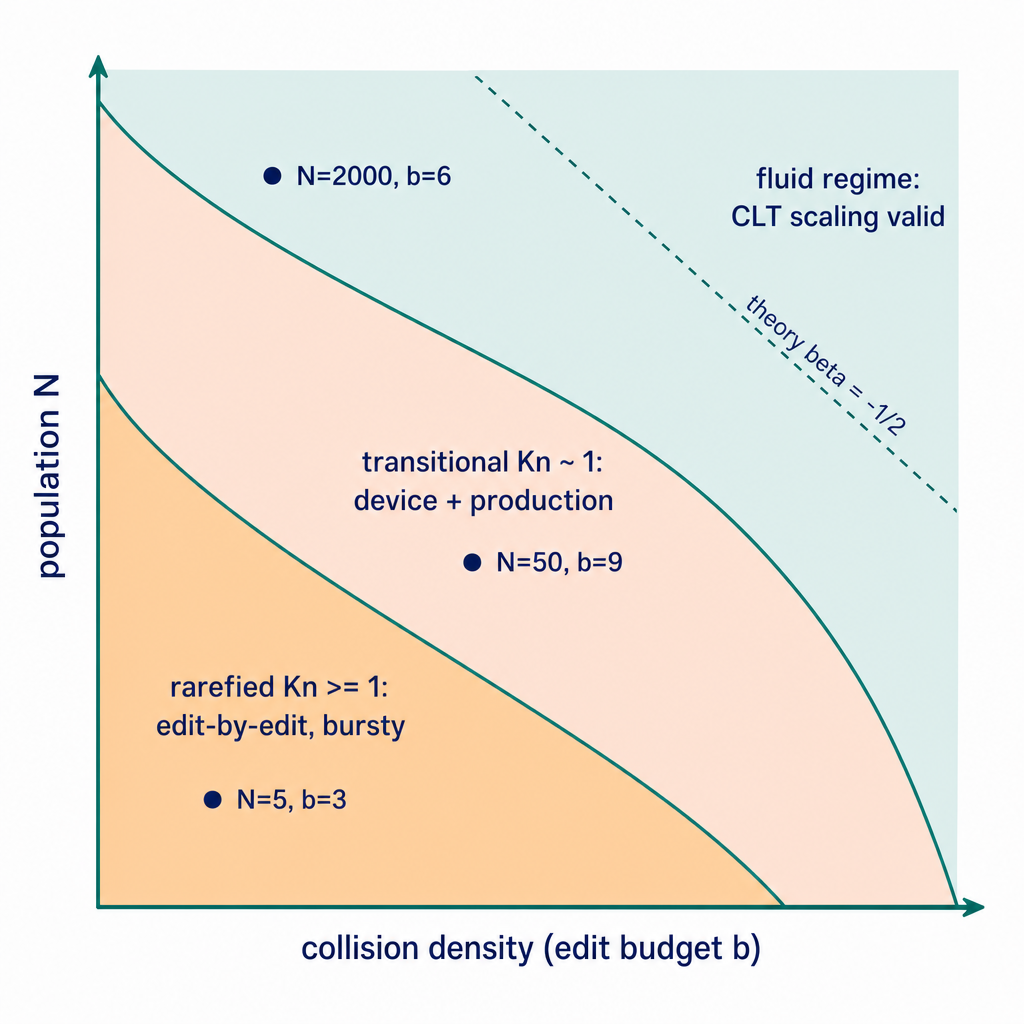}
\caption{The regime map. Collision density (edit budget) against population $N$; the device points mark where our measurements sit. The CLT signature $\beta=-1/2$ is a property of the fluid corner; the rarefied corner is edit-by-edit and bursty; the device operates in the transitional band, which is why signatures are density-gated (\S\ref{sec:experiments}). The placement of the boundaries on the $b{\to}\Kn$ axes is heuristic.}
\label{fig:regime-map}
\end{figure}

\begin{theorem}\label{thm:meanfield}
A plugin Markov-jump population with i.i.d.\ initial conditions and a bounded collision kernel (no dominating instance) converges, under propagation of chaos in the McKean--Vlasov sense, to a deterministic nonlinear kinetic equation of the Boltzmann type with the same four operators.
\end{theorem}
The derivation route here is mean-field~\cite{sznitman1991,kac1956}. The Boltzmann--Grad alternative~\cite{lanford1975,deng2024longtime} requires control of recollisions, which the plugin gas introduces by design. The applicability of the resulting closure has a structural basis as well. It is cleaner here for a deeper reason. The plugin state is \emph{enriched with its collision history} through the credit record $e$, and the experience database is precisely this stored collision history, so the correlations a gas loses are partially materialized as state, a heuristic reading; the residual global coupling through the shared database enters the kernel's task-dependence rather than hidden pair correlations. The no-dominating-instance premise is measured on the device, not assumed (Appendix~\ref{app:ops}).

\begin{proposition}\label{prop:plateau}
Under H2 and strict monotone waterline, a plateau implies exactly one of: (a) proposal exhaustion, meaning the acceptance rate tends to zero; or (b) information exhaustion, meaning the mean score gain on the gate's own validation set stays positive, $\delta_{\mathrm{gate}} > 0$, while the mean gain measured on a disjoint task set vanishes, $\delta_{\mathrm{ind}} \approx 0$. Distinguishable from logged data alone.
\end{proposition}

\begin{proposition}\label{prop:fluct}
Under Theorem~\ref{thm:meanfield} and absolutely summable autocorrelation ($\sum_k|\rho_k| < \infty$), the run-to-run standard deviation of the population score satisfies
\begin{equation}
\mathrm{std} \;\approx\; \sqrt{\tfrac{T}{N_{\mathrm{eff}}}\bigl(1 + 2{\textstyle\sum_k} \rho_k\bigr)}.
\label{eq:fluct}
\end{equation}
\end{proposition}
Prediction: $\beta = -1/2$ on log-log; a contrary $\beta$ indicates rarefaction or over-correlation (Fig.~\ref{fig:regime-map}); the shared database can violate summability.

\begin{theorem}\label{thm:demon}
Within the certificate's assumptions, separating the frozen validation set from the evolving set is \emph{necessary}: merging them voids the drift hypothesis, hence the bound. The Maxwell daemon operates outside the gas. \emph{Scope}: non-separated designs may still converge, but they simply exit the certificate's coverage.
\end{theorem}
The argument runs on the drift premise: the gate certifies strictly improving post-states only on the set it measures; if the measured set is itself evolving, acceptance is self-referential and the per-round drift behind the certificate cannot be established. The formalized reservoir proves this necessity direction explicitly: non-separability voids the certificate premises (Appendix~\ref{app:lean}).

\begin{figure}[H]
\centering
\includegraphics[width=0.50\textwidth]{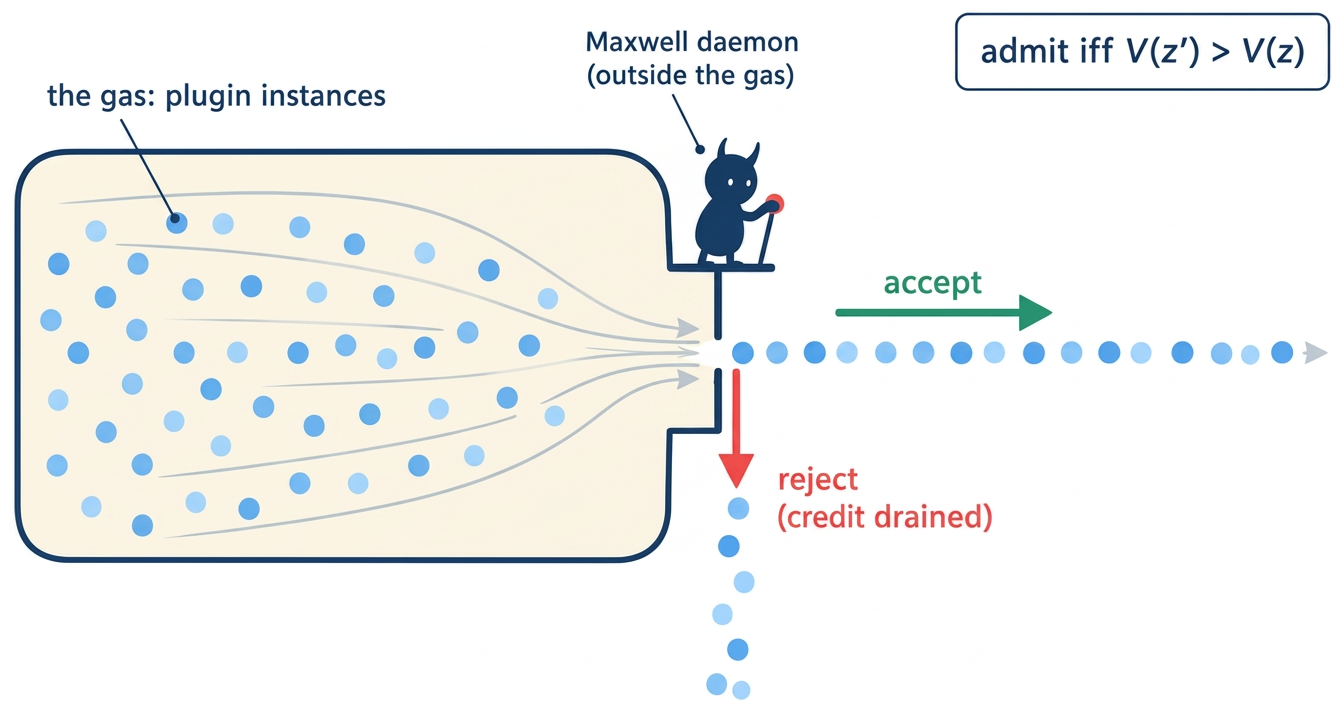}
\caption{The gate as a Maxwell-daemon valve outside the gas: candidates stream past the slit and are measured on the held-out set; strictly improving post-states are admitted, the rest are rejected with credit drained. The certificate requires the daemon to sit outside the flow (Theorem~\ref{thm:demon}).}
\label{fig:gate}
\end{figure}

\subsection{Thermodynamics and transport}\label{sec:ness}

The collision rule sampled by the LLM is asymmetric and carries no thermostat of its own: without the gate, the population does not thermalize but drifts to the LLM's \emph{prior modality}, the kinetic form of model collapse. All training signal is therefore supplied by the gate, and the effective learning rate separates into
\begin{equation}
\eta_{\mathrm{eff}} \;=\; f_{\mathrm{edit}}\; I_{G},
\qquad
I_{G} \;=\; I\big(A;\,\Delta V^{*}\big) \;=\; \mathbb{E}\, D_{\mathrm{KL}}\!\big(P_{\Delta V^{*}\mid A}\,\big\|\,P_{\Delta V^{*}}\big),
\label{eq:learn}
\end{equation}
with $f_{\mathrm{edit}}$ the rate at which the reflector proposes edits, $A$ the gate's accept--reject adjudication, and $\Delta V^{*}$ the score gain measured on a task set disjoint from the evolving population. $I_{G}$ is the mutual information between the adjudication and the gain: the expectation, taken over adjudications, of the KL divergence from the gain distribution conditioned on the adjudication, $P_{\Delta V^{*}\mid A}$, to the unconditional distribution $P_{\Delta V^{*}}$. Edits alone move the population; the gate decides whether the motion counts, and $I_{G}$ vanishes exactly in the information-exhaustion channel of Proposition~\ref{prop:plateau}.

With the gate active, reward injection (drive) and credit drainage (dissipation) stay balanced, and stationarity of Eq.~\eqref{eq:ktp} pairs transport against interactions,
\begin{equation}
\underbrace{\;\upd{v}\cdot\nabla_x f \;+\; \nabla_{\upd{v}}\cdot(F\,f)\;}_{\text{flow, field}}
\;=\;
\underbrace{\;Q[f,f] \;+\; R[f] \;+\; G[f]\;}_{\text{collisions, reactions, gate}},
\label{eq:ness}
\end{equation}
a non-equilibrium steady state (NESS): convergence is relaxation to it, a plateau is vanishing entropy production, and escaping a plateau requires changing parameters (temperature, task field, gate strictness), not applying force. A fixed validation set makes the daemon efficiency shrink, $\delta_{\mathrm{gate}}>0$ while $\delta_{\mathrm{ind}}\approx 0$: molecular chaos failing at the daemon--gas level---the device-level manifestation of $I_{G}\to 0$ in Eq.~\eqref{eq:learn}. The full thermodynamics, from the H-functional and the information-daemon boundary to the Knudsen regime map and the gas--liquid--solid phase picture, appears in Appendices~\ref{app:labels}--\ref{app:ops}.


Eq.~\eqref{eq:ktp} also admits a transport reading through its moments,
\begin{equation}
\rho \;=\; \int f\,dv\,de, \qquad
\rho\,u \;=\; \int v\,f\,dv\,de, \qquad
\tfrac{3}{2}\,\rho\,T \;=\; \int \tfrac{1}{2}\,|v-u|^{2}\,f\,dv\,de,
\end{equation}
namely the plugin density or coverage, the improvement drift, and the exploration temperature; the gate enters as the external selection field, the daemon outside the gas, giving the system its direction.

The analogy is bounded. The conservation ledger shows total impulse $\Sigma v$ is \emph{not} conserved, so no Navier--Stokes closure exists, and we claim no literal turbulence, vortex streets, Magnus force, or Bernoulli law. The transport coefficients name measurable quantities---the retrieval radius as $\kappa$, the edit damping as $\mu$, the sampling noise as $D$---and every regime name is an analogy for a measured index, not a derived constant; the full defense, with the Euler-type closures and the complete license inventory, appears in Appendices~\ref{app:labels}--\ref{app:ops}.

\section{Experiments}\label{sec:experiments}

The experiments test the theory's signatures on a minimal WebShop~\cite{yao2022webshop} training loop. Every measurement is a \emph{signature test} at one of three license levels: \emph{L0 architectural correspondence}, \emph{L1 own-equation prediction}, and \emph{L2 borrowed name}. Four pre-registered groups carry the load: fluctuation scaling across population sizes and edit densities (Proposition~\ref{prop:fluct}); the collision channel, probed by fan-3 retrieval against fan-1 controls; NESS relaxation in 500-round runs (Section~\ref{sec:ness}); and the gate-versus-independent score split of Proposition~\ref{prop:plateau}. Radius and edit-budget sweeps complete Table~\ref{tab:headline}. As the first kinetic theory of this loop, comparisons are against pre-registered criteria and within-design controls, not external methods.

The environment is WebShop text mode: 1{,}000 products, 6{,}910 synthetic-attribute goals, in-process BM25 simulation, search/click/end actions, reward 1 for an attribute-and-price-matched purchase; CPU, one public LLM API, seeded and reproducible (operator instantiation and full protocol: \S\ref{sec:device}).

\subsection{The faithful minimal instance}\label{sec:device}

To make the equation measurable, we instantiate each operator on the smallest system that preserves the loop's structure. The correspondences are licensed \emph{L0 structural}, justified by the code architecture rather than analogy: instance$=$particle, co-activation$=$collision, retrieval$=$interaction range, gate$=$daemon, reward$=$injection, and rejection$=$dissipation.

A plugin is the 4-field parameter vector $\theta = $(search rank, read policy, price floor/ceiling); a version instance adds a tag profile (spaCy nouns of its seeded tasks), version counter, and usage history---mirroring production registry parameters (discrete scalars). Retrieval scores instances by tag Jaccard, $\exp(-\alpha d)$, sampling from the top-$k$; edits are $\le L$ schema-validated single-field changes proposed by the reflector from a contrastive feedback table; the gate replaces the population only on strict held-out improvement---the candidate's validation score must exceed the incumbent's, both measured on the same held-out set (20 evaluation episodes per candidate, ties rejected; the incumbent baseline is the initial library's validation score $v_0$). Full specification: Appendix~\ref{app:device}.

The protocol is frozen before any run: $N_{\mathrm{eff}}$ := population size $N$ (occupied-grid interpretation); arms $N \in \{5,50,500\}$ $\times$ seeds $\{11,21\}$; retrieval fixed at $(k{=}3,\alpha{=}1)$ during evolution; radius sweep on a frozen $N{=}50$ library over $\alpha \in \{0.25{\ldots}4\}$ and $k \in \{1,2,3,5\}$, 100 tasks per cell; the parallel seed pair fixed in advance; no parameter changed after results are seen (full protocol in the drivers' frozen headers and run metadata).

\subsection{Fluctuation scaling in population size}\label{sec:exp-p1}

Prediction (\emph{L1 own-equation}): run-to-run dispersion of the population score decays with $N$ ($\beta < 0$ on log-log).

The population-size arms (six runs, $N \times 2$ seeds, $(k{=}3,\alpha{=}1)$, 10 rounds; all finals validation measurements) give finals (seed 11/21) $0.203/0.292$, $0.216/0.267$, $0.244/0.172$ ($N{=}5/50/500$); accepts $2/6$, $2/3$, $2/2$ (per-arm $v_0$ in Appendix~\ref{app:records}). Under this baseline the $N{=}500$ dilution-starvation disappears (2 accepts both seeds); the fitted slope is flat, with per-arm std $0.044 \to 0.025 \to 0.037$ and $\hat\beta = -0.04$ ($r^2 = 0.12$, 3 points, $J{=}2$). \emph{Verdict: the fluctuation signature is not detectable at this power}, neither as the predicted $-1/2$ nor as any trend. Table~\ref{tab:headline} collects this and the remaining headline probes with their pre-registered criteria and verdicts.

\begin{table}[H]
\centering
\scriptsize
\setlength{\tabcolsep}{4pt}
\caption{Pre-registered headline probes (provenance: Appendix~\ref{app:records}, which records every pre-registered probe ID including P3 and P5). Hatted quantities are sample estimators of the $\delta$ of Proposition~\ref{prop:plateau}.}
\label{tab:headline}
\begin{tabular}{@{}llll@{}}
\toprule
Probe & Criterion & Measured & Verdict \\
\midrule
P1 scaling, $b{=}3$ & $\hat\beta < 0$ & $-0.04$ ($r^2{=}0.12$) & flat; undetectable at $J{=}2$ \\
P1 dial, $b{=}9$ & report as measured & $-0.58$ ($r^2{=}0.94$) & matches $-1/2$, directional \\
P2 radius peak & peak at finite $\alpha$ & $\alpha{=}2$: $0.1684$ & supported, directional \\
P4 Goodhart & $\widehat\Delta_{\mathrm{gate}}{>}0$, $\widehat\Delta_{\mathrm{ind}}{\approx} 0$ & $+0.0064$ vs.\ $-0.0486$ & fires (channel b) \\
F1 laminar & additive post-transient window & NESS plateau at 500r & window open at $N{=}2000$ \\
Collision & fan-3 vs fan-1 final & $0.495$ vs $0.301$ & net gain $+0.19$ \\
$\varepsilon$ phase & monotone-smooth response & $2.7{\to}0.7$, max drop $1.0$ & smooth; no jump (bar $2.0$) \\
\bottomrule
\end{tabular}
\end{table}

\begin{figure}[!t]
\centering
\includegraphics[width=\columnwidth]{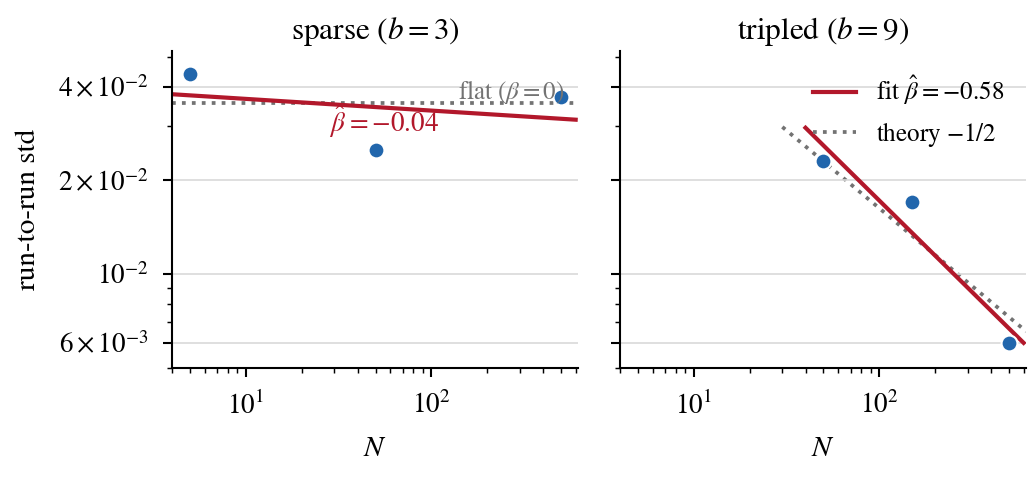}
\caption{Density-gated scaling (corrected gate). Points: per-arm standard deviations of the two seed finals (records: Appendix~\ref{app:records}). Left: sparse density, flat $\hat\beta=-0.04$. Right: tripled density, CLT emerges at $\hat\beta=-0.58$ vs.\ $-1/2$.}
\label{fig:density-gated}
\end{figure}

The density dial (formal8: $b{=}9$, 2 seeds, 30 rounds) reads the regime clause on the device: the CLT signature $\beta=-1/2$ appears once collision density suffices; the $b{\to}\Kn$ mapping is a \emph{heuristic}, so the slope is reported as measured. Per-arm std across $N \in \{50,150,500\}$: $0.023 \to 0.017 \to 0.006$ (finals in Appendix~\ref{app:records}); $\hat\beta(b{=}9) = -0.58$, $r^2 = 0.94$ (Fig.~\ref{fig:density-gated}). With the flat $b{=}3$ reading: \emph{sparse density buries the signature; tripled density brings it out at the predicted magnitude} ($-0.58$ vs.\ $-1/2$; directional).

The collision channel is probed by formal9: a no-op line had pinned the retrieval fan to one instance per episode, so ``co-activation $=$ collision'' was only structural. Under the fixed harness, a fan-3 arm (credit to the best of three instances) reaches mean final $0.495$ vs.\ $0.301$ for the fan-1 control (Appendix~\ref{app:records}): \emph{co-activation brings a net gain} of $+0.19$, so the collision channel is measurably load-bearing.

The scale probes (formal11/12, corrected gate) ground the saturation reading of the certificate: their 24 skill groups traverse only about $10\%$ of the expected hitting interval, so late in training the binding constraint is the resolution floor, not the skill.  formal11 ($N{=}2000$, $b{=}6$, 150 rounds): $v_0{=}0.184 \to 0.274$, 8 accepts, and a per-accept fit of $\hat\mu{=}0.014$ with $r^2{=}0.88$, so entrance and plateau persist at this dilution. formal12 ($N{=}50$, $b{=}9$, 500 rounds): $v_0{=}0.208 \to 0.465$, 8 accepts, tail $0.003$/accept, the pre-registered NESS discriminator fires at $r^2{=}0.80<0.9$; \emph{convergence is observed as NESS relaxation}, and the pre-correction ``entrance ${>}500$ rounds'' bound is retired (Appendix~\ref{app:records}).

Across the pre-registered probes, six of the eight L1 signatures end supported or observed, NESS relaxation included; the probed \emph{L2}-name candidates fail as observables---Bernoulli, K\'arm\'an vortex streets, gravity/acoustic waves---and the Magnus correlation fires only gate-dependently, licensing no circulation claim (formal13b). Criteria are frozen in the released drivers' headers, and the supported/borderline/unsupported triage is a post-hoc layer over pre-registered observables and thresholds; Tables~\ref{tab:license-l0}--\ref{tab:license-boundary} (Appendix~\ref{app:labels}) give the complete inventory, and a corrected-gate verdict summary appears in Table~\ref{tab:verdicts}.

\section{Conclusion and Discussion}\label{sec:discussion}

The gated self-evolution loop admits a statistical theory scoped by four architectural axioms: a machine-verified hitting-time certificate, a measured signature set with its license inventory, and an analogy boundary located by the conservation ledger. Device records are small (two seeds per arm, one model, one task family); the certificate is instance-independent, the measurements are not. Under the corrected instrument the density-gated readings stand, with a near-additive waterline and CLT scaling emerging at tripled edit density, and the remaining windows stay open under pre-registered probes. As a navigation instrument, the theory prices validation budgets, flags silent collapse, and marks which designs exit its coverage.

\begin{AIdeclaration}
Generative AI (MiniMax-M3 and AI coding assistants) was used in the following roles, all disclosed: (1) as the \emph{reflector} inside the experimental device---proposing bounded skill edits under schema validation and the held-out gate (an \emph{object of study}, not a writing aid); (2) for code development (device and analysis boilerplate); (3) for experiment execution under author direction and the pre-registered protocols; (4) for reference discovery and for drafting/polishing of the manuscript under author direction---all cited references were manually verified against the sources, and all scientific content, theorem statements, and conclusions were decided and reviewed by the authors; (5) for Lean proof development assistance---all proofs are machine-verified (sources in the supplementary).
\end{AIdeclaration}

\begin{Reprostatement}
All experiments run on CPU with one public API (MiniMax-M3) and a public environment (WebShop). The device (ktp/), the pre-registered drivers (run\_formal*.py, decision criteria frozen in headers), and the raw artifacts (runs/*.json) are released in the supplementary; every number in Section~\ref{sec:experiments} is reproducible from them; seeds are in the run metadata; Lean sources for Theorem~A are included. One limitation: the reflector calls a hosted LLM whose behavior may drift between versions---the released artifacts are the ground truth for all reported numbers, and the drivers can be re-run against any compatible endpoint.
\end{Reprostatement}

\bibliographystyle{iclr2027_conference}
\bibliography{references}

\clearpage
\setcounter{page}{1}
\renewcommand{\thepage}{A\arabic{page}}
\appendix
\begin{center}
{\large\scshape Appendices}
\end{center}
\vspace{0.6em}
\section{Lean Verification Record}\label{app:lean}

Theorem hierarchy A--E, verified in the released Lean sources---zero sorry, full build green on 2026-08-20 (8{,}712 jobs; Lean 4.33 + Mathlib 4.33), grep audit clean:
\begin{itemize}\itemsep1pt
\item \emph{(A) Hitting-time certificate.} 12 theorems + 2 auxiliary lemmas; the expectation bound is machine-checked, and hitting-time measurability is proven, not assumed.
\item \emph{(D) Fluctuation scaling.} the geometric autocorrelation limit $(1+\phi)/(1-\phi)$ is Lean-verified.
\item \emph{Support facts.} the Henderson/Carnahan--Starling contact-value leading term $5/2$~\cite{carnahan1969} (slope form, Lean-verified); algebraic-cooling break below the H\'ajek logarithmic guarantee line~\cite{hajek1988cooling} (Lean-verified).
\item \emph{Gate monotonicity.} the gate score under the strict gate is monotone (Theorem 5 of the reservoir).
\item \emph{(E) Separation necessity.} cited from the formalized reservoir Theorem 7 (separability is necessary); the necessity direction is what makes the certificate conditional on the design axiom.
\end{itemize}
Build instructions ship with the released sources.

\subsection{Proof of Theorem~\ref{thm:hitting}}

Write $V^{*}_t$ for the deployed (best accepted) validation score after round $t$, and let $V_0$ denote the initial validation gap, $V_0 = V^{\max} - V^{*}_0$, with $V^{\max}$ the score cap; $T_{\varepsilon}$ is the first round at which the gap falls to $\varepsilon$ or below. Under H1$'$, whenever the gap has not closed there exists a strictly improving edit with gain ${\ge}\Delta$, the reflector proposes one with probability ${\ge}p_{\min}$, and under H2--H3$'$ every accepted proposal enters the deployed lineage while a rejected proposal changes nothing. Hence for $t < T_{\varepsilon}$ the conditional expected gap reduction per round is at least $p_{\min}\Delta$,
\[
\mathbb{E}\big[\,G_t - G_{t+1} \,\big|\, G_1,\dots,G_t\,\big] \;\ge\; p_{\min}\Delta ,
\]
with $G_t$ the gap at round $t$. Validation scores are bounded, so the per-round reduction is uniformly bounded and optional stopping applies; summing the drift,
\[
p_{\min}\Delta\;\mathbb{E}[T_{\varepsilon}] \;\le\; G_0 \;=\; V_0 ,
\]
which is the stated bound (for the threshold version, $\mathbb{E}[T_{\varepsilon}] \le (V_0-\varepsilon)/(p_{\min}\Delta)$). Monotonicity of the best-observed score along the accepted lineage is immediate from H2--H3$'$: the deployed version changes only on a strict improvement. Robustness to a change of state space holds because the argument uses only the per-round drift floor, not the coordinates in which scores are represented. The derivation is machine-checked in the released Lean sources. H1$'$ itself is a design axiom rather than an estimable quantity---the existence of an unexplored improving edit is not decidable from logs; what the device measures is the realized drift of the accepted lineage---operationally, the acceptance rate times the mean accepted gain---and the certificate is applied post hoc as a consistency check on that realized drift: when the realized drift falls below the resolution floor, the binding constraint is measurement resolution, not training time (the saturation reading of Section~\ref{sec:exp-p1}).

\subsection{Theorem~\ref{thm:demon} in the formalized reservoir}

The formalized reservoir proves the necessity direction explicitly (Theorem 7): its model admits measurement sets that overlap the evolving population, and within that model non-separability breaks the drift premise---the improvements the gate registers then need not track the out-of-sample objective, so the certificate's premises are void. Separability is therefore not a convenience but the condition of applicability: the full statement and machine-checked proof are in the released sources.

\subsection{Derivation route for Proposition~\ref{prop:fluct}}

Under Theorem~\ref{thm:meanfield} the population is an exchangeable ensemble of $N_{\mathrm{eff}}$ effectively independent contribution streams, and the summed score over $T$ rounds is a stationary sequence with absolutely summable autocorrelation $\rho_k$. Its variance takes the Green--Kubo form $\mathrm{Var} = \tfrac{T\sigma^2}{N_{\mathrm{eff}}}\bigl(1 + 2\sum_k \rho_k\bigr)$, whence $\mathrm{std} \propto N_{\mathrm{eff}}^{-1/2}$ and $\beta = -1/2$ on log-log; the geometric-autocorrelation special case $(1+\phi)/(1-\phi)$ is verified in the released Lean sources. A failure of decay under adequate power is the diagnostic for excessive inter-particle correlation.

\subsection{Remark on Proposition~\ref{prop:plateau} (exhaustiveness)}

Under the strict gate (H2), an acceptance is exactly an event with positive held-out gain, so persistent acceptances coincide with $\delta_{\mathrm{gate}}>0$: the middle case (acceptances persist while $\delta_{\mathrm{gate}}$ vanishes) is empty by the definition of acceptance---nor can gains ${\to}0^{+}$ reopen it: any positive mean gain still satisfies $\delta_{\mathrm{gate}}>0$, which is channel b's own clause. The dichotomy is therefore exhaustive given H2; what H2 buys here is one more instance of the gate's design carrying the theory.

\section{Device Specification}\label{app:device}

The environment is WebShop text mode, in-process Flask simulation with a pure-Python BM25 search index (no browser, no Java). Corpus: 1000 products; task goals are generated synthetically from item attributes (6{,}910 goal instances; human-instruction goals cover only 13 of the 1000 products and are not used). Episode mechanics: the agent sees the current page text, can search[keywords], click[value] (buttons, links, options), end; reward $=1$ when the bought item matches the target within the price bound, graded partial otherwise. Determinism: task index $\to$ goal via the simulator's session index; with a seeded RNG every rollout is reproducible.

A plugin is the 4-field parameter vector $\theta = (\text{search\_rank},\ \text{read\_policy},\ \text{price\_floor},\ \text{price\_ceiling})$; a version instance is $(\theta,\ \text{tag profile},\ \text{version}\ \#,\ \text{usage history})$. The tag profile is built from task-query words (spaCy nouns) of the task the instance was seeded on; the retrieval distance between a task and an instance is $1-\mathrm{Jaccard}(Q_{\mathrm{task}},\ \mathrm{tags}_i)$.

Retrieval scores instances by $\mathrm{score}_i = \exp(-\alpha\, d_i)$; candidate set = top-$k$ by score; the executed instance is sampled from that set with weight proportional to score (seeded RNG). $(k,\alpha)$ are the two collision-radius knobs.

Each round, the reflector (MiniMax-M3) receives the per-instance contrastive table: use counts, per-task rewards, mean reward, trailing reward lists, and parameter vectors; it proposes at most $L$ field edits targeting the failing instances (one field per edit, values from the legal grid; schema-validated, an invalid proposal is discarded wholesale). LLM failures fall back to a deterministic heuristic so the harness never blocks.

The gate evaluates the candidate library on a held-out task set (20 episodes); replacement requires strict improvement of the held-out mean reward over the incumbent (ties rejected). Rejected candidates leave no trace in the population, and their evidence feeds the next reflection.

\section{Experiment Records}\label{app:records}

\subsection{Population-size arms (Prediction 1, corrected gate)}

\begin{tabular}{cccccc}
\toprule
$N$ & seed & $v_0$ (val) & $s_{0,\mathrm{train}}$ & waterline (accepted) & accepts \\
\midrule
5 & 11 & 0.0373 & 0.1160 & 0.1924, 0.2033 & 2 \\
5 & 21 & 0.1000 & 0.2569 & 0.1182, 0.1880, 0.2237, 0.2346, 0.2402, 0.2921 & 6 \\
50 & 11 & 0.2083 & 0.0830 & 0.2090, 0.2158 & 2 \\
50 & 21 & 0.1425 & 0.1861 & 0.1894, 0.2439, 0.2665 & 3 \\
500 & 11 & 0.1478 & 0.2708 & 0.1597, 0.2444 & 2 \\
500 & 21 & 0.0507 & 0.1455 & 0.1257, 0.1715 & 2 \\
\bottomrule
\end{tabular}

Regression fit (3 arms): $\hat\beta = -0.043$, $r^2 = 0.12$ (Fig.~\ref{fig:fluct}); per-arm std of finals $0.044$ ($N{=}5$), $0.025$ ($N{=}50$), $0.036$ ($N{=}500$)---flat. Note the train--val baseline gap that motivated the correction is visible per arm (e.g.\ $N{=}5$/s21: train $0.257$ vs.\ val $0.100$): under the old mixed-set gate this arm's incumbent sat at $0.257$ and only one edit ever cleared it; under the honest bar the same arm accepts six edits and reaches $0.292$. \emph{Provenance:} formal8\_p1\_*\_run.json; the pre-correction table (train-baselined, $\hat\beta=+0.178$) is retained in formal\_p1\_*\_run.json and is superseded. Seed-level uncertainty at $J{=}2$ is uninformative by construction: the slope standard error is $0.12$ with one degree of freedom per fit, so interval estimates span the admissible range. The pre-registered judgment in the driver header ($|\hat\beta+1/2|<0.15$, an exploratory direction test at 3 points $\times$ $J{=}2$) records the outcome as a deviation---the power statement, not a confidence interval, is the correct summary at this $J$. A seed-level bootstrap is degenerate at $J{=}2$ (each arm's standard error already uses both seeds, so resampling has no room to vary), which is why the analytic standard error rather than a bootstrap interval is reported. The power analysis itself was pre-registered in the experiment plan: with arms spanning 2 orders of magnitude in $N$, $J{=}10$ gives slope SE ${\approx}0.17$ (4$\sigma$ discrimination of $-1/2$ vs.\ 0) and $J{=}6$ is marginal ($2.2\sigma$); the executed $J{=}2$ (quota decision recorded 2026-08-28, plan v2.1) falls below both, and the plan's pre-registered fallback---downgrade to an order-of-magnitude consistency check---is exactly the verdict reported. The plan ships with the supplementary.

\subsection{Radius sweep (Prediction 2)}

$\alpha$: $0.25\to0.1448$, $0.5\to0.1491$, $1.0\to0.1569$, $2.0\to0.1684$, $4.0\to0.1680$. $k$ (at $\alpha{=}1$): $1\to0.1630$, $2\to0.1559$, $3\to0.1569$, $5\to0.1566$. Frozen $N{=}50$ library (seed 11 initial population), 100 tasks per cell, 15s per cell.

\subsection{Closure replay (Prediction 3)}

\emph{A worked collision-gate round} (concrete instance of the loop of Section~\ref{sec:operators}, $N{=}5$, seed 11, corrected gate): the round-0 incumbent is the initial library's validation score $v_0 = 0.0373$. At round 1 the reflector proposes one bounded edit to the weakest instance, the gate measures the candidate on the 20-episode held-out set, the score rises $0.0373 \to 0.1924$ ($\Delta V = +0.155$), and the edit is kept. Rounds 2--5 propose further edits; each candidate measures below the incumbent and the gate rejects all four---the daemon's dissipation, visible in the log as zero $\Delta V$. Round 6 accepts ($+0.011$); rounds 7--10 reject. Ten rounds move the waterline from $0.0373$ to $0.2033$ with exactly two gate admissions; the trajectory is monotone by construction.

Additivity replay on the corrected-gate runs with $\ge 3$ watermark points: $N{=}5$ seed 21 (7 points): $\hat\mu = 0.031$, $\hat v_0 = 0.107$, $\max|\mathrm{dev}| = 0.024$; $N{=}50$ seed 21 (4 points): $\hat\mu = 0.043$, $\hat v_0 = 0.147$, $\max|\mathrm{dev}| = 0.012$. Under the honest baseline the waterline is much closer to additive than the pre-correction data suggested (deviations ${\le}0.024$ vs.\ up to $0.32$ before): the large first-burst concentration previously attributed to non-additivity was mostly the train--val baseline gap.

\subsection{Goodhart probe (Prediction 4)}

Gate deltas pooled over 6 runs: $\widehat\Delta_{\mathrm{gate}} = +0.0064$; independent-set deltas (final population on a fresh 60-task segment): $\widehat\Delta_{\mathrm{ind}} = -0.0486$. Diagnostic: channel (b) (Goodhart precursor). Alternative readings are bounded by construction: (i) evaluation noise cannot produce opposite signs at these magnitudes, since both deltas are means over the same 20-episode budget (gate arm over accepted events, independent arm over a fresh 60-task segment); (ii) task-distribution shift is bounded by the generator---both sets are synthetic-attribute goals drawn from the same corpus; (iii) regression to the mean is not a confounder here but the effect itself: the gate delta is measured on the selection set and the independent delta out of sample, which is precisely the selection-induced inflation the Goodhart reading names.

\subsection{Corrected-gate re-measurement (formal8)}

All three blocks on the corrected instrument (incumbent $=$ validation $v_0$; pre-registered in the driver header): \emph{P1-fix} ($b{=}3$, 6 runs) --- see the table above; $\hat\beta = -0.043$, $r^2 = 0.12$. \emph{Dial-fix} ($b{=}9$, 30 rounds, 6 runs): finals (s11/s21) $0.345/0.391$ ($N{=}50$), $0.250/0.217$ ($N{=}150$), $0.274/0.261$ ($N{=}500$); accepts $4/4$, $1/1$, $3/6$; per-arm std $0.0232/0.0165/0.0062$; $\hat\beta = -0.576$, $r^2 = 0.94$ (slope standard error $0.15$ at $J{=}2$; the pre-registered dial judgment reports the slope as measured, with no convergence-to-$-1/2$ presumption---driver header). \emph{F-Drive-fix} ($N{=}500$, $b{=}6$, 60 rounds): $v_0 = 0.1478 \to 0.2444$, 2 accepts, lift $+0.0966$ (below the pre-registered $+0.10$ bar: marginal). Wall 15056\,s total.

\begin{figure}[h]
\centering
\includegraphics[width=0.5\textwidth]{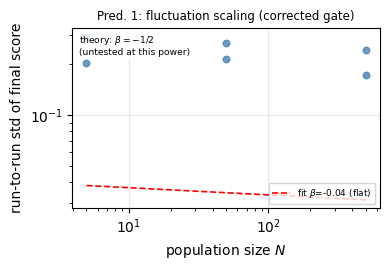}
\caption{Prediction 1 under the corrected gate: six endpoints, flat fit $\hat\beta=-0.04$---no detectable $N$-scaling at $J{=}2$.}
\label{fig:fluct}
\end{figure}

\subsection{Collision-channel probe (formal9)}

$N{=}50$, $b{=}6$, 30 rounds, seeds \{11,21\}, fan-3 arm vs.\ fan-1 control: finals $0.529/0.461$ (fan-3, mean $0.495$) vs.\ $0.309/0.292$ (fan-1, mean $0.301$) (Fig.~\ref{fig:device-traces}); accepts $5{+}2$ vs.\ $5{+}4$; per-episode credit spread over $13/11$ vs.\ $17/14$ documents; wall 6857\,s. Pre-registered descriptive verdict: co-activation brings a net gain ($+0.194$).

\begin{figure}[h]
\centering
\includegraphics[width=0.48\textwidth]{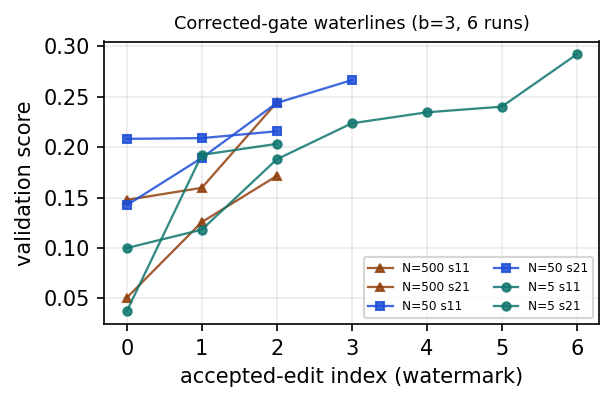}\hfill
\includegraphics[width=0.48\textwidth]{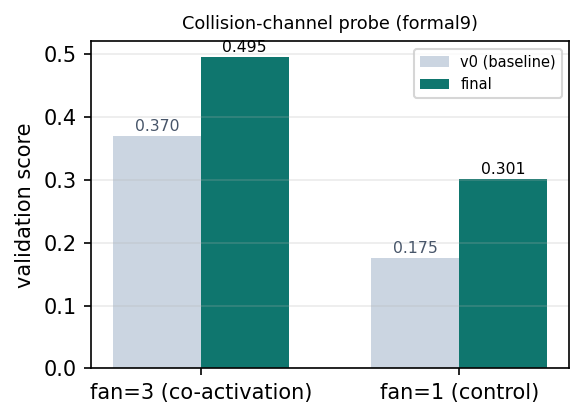}
\caption{Device traces. Left: six $b{=}3$ validation waterlines (near-additive, deviations ${\le}0.024$). Right: collision probe---fan-3 final $0.495$ vs.\ fan-1 $0.301$ (formal9).}
\label{fig:device-traces}
\end{figure}

\subsection{Gate-strictness sweep (formal10, corrected gate)}

$N{=}50$, $b{=}3$, $\varepsilon \in \{0,0.005,0.01,0.02,0.05\}$ $\times$ seeds $\{11,21,31\}$, 10 rounds: acceptances per seed $[2{+}3{+}3,\;2{+}3{+}3,\;0{+}3{+}4,\;0{+}3{+}3,\;0{+}1{+}1]$ (means $2.7,2.7,2.3,2.0,0.7$); finals per $\varepsilon$ in formal10\_result.json. Pre-registered jump bar (max adjacent drop $\ge 2.0$) not met: \emph{smooth control-parameter response}; the qualitative monotone-smooth reading of the pre-correction sweep survives the corrected instrument.

\subsection{Scale probes under the corrected gate (formal11/12)}

\emph{Giant scale, re-run} (formal11 $=$ formal6's protocol under the corrected gate; $N{=}2000$, $b{=}6$, fan as-run 1, 150 rounds, seed 11): $v_0 = 0.184$ (validation), final $0.2738$, 8 accepts (waterline $0.184 \to 0.2738$); post-transient per-accept fit (skip 4, 5 points): $\hat\mu = 0.0141$, $r^2 = 0.878$---the pre-registered laminar-window criterion ($R^2{>}0.95$) is not met and the $R^2{<}0.9$ branch records the outcome as measured: \emph{entrance/plateau persists at this dilution within 150 rounds}; wall 8328\,s. \emph{Dense long run, re-run} (formal12 $=$ formal7's protocol under the corrected gate; $N{=}50$, $b{=}9$, fan as-run 1, 500 rounds, seed 11): $v_0 = 0.2083 \to$ final $0.4651$, 8 accepts; fit $\hat\mu = 0.047$, $r^2 = 0.797$; last-third slope $0.0029$/accept---the pre-registered NESS discriminator ($r^2{<}0.9$) fires: \emph{the plateau is reached within 500 rounds}, convergence as relaxation evidence; wall 15496\,s. Per-accept increments compare as $0.014$ ($N{=}2000$) vs.\ $0.047$ ($N{=}50$): dilution thins the value channel, not the accept count.

\subsection{Signature re-runs under the corrected gate (formal13a/13b)}

Controlled copies of the formal4/5 drivers with criteria logic verbatim and fan as-run 1 (the gate correction is the single changed variable; the fan effect is formal9's). \emph{A-group} (formal13a, wall 22617\,s): F-Turbulence CVs $0.39/0.78/0.87$ (mean $0.68 > 0.6$: supported); F-Transition CV($\alpha$) $1.24, 0.95, 0.97, 0.20, 0.51$---non-monotone, no transition (borderline, matching judge\_rederived); F-Damping with a deterministic proposer still applies edits (1 accept; waterline $0.224 \to 0.252$, no decay)---the no-drive design remains confounded, \emph{withdrawal maintained}; F-Drive $0.148 \to 0.244$, 2 accepts, lift $+0.097$ (marginal; reproduces formal8's F-Drive-fix exactly); F-Conduction cross-seed spread $0.031$ ($\alpha{=}0.25$) vs.\ $0.007$ ($\alpha{=}4$): supported; F-Brownian waterline $0.223 \to 0.241 \to 0.308$, monotone: supported. \emph{B-group} (formal13b, wall 9802\,s): Bernoulli density$_A = 0.103$, density$_B$ undefined (not supported); Coanda $41.2\%$ vs.\ $30\%$ baseline (supported); K\'arman 0 alternators (not supported); Magnus $r = 0.638$ ($n{=}14$; the pre-registered $r{>}0.3$ criterion fires---flipped from $r{=}0.04$; see the license table: the classical mechanism remains absent); waves near/far drop $0$ (not supported); tides $\mathrm{ac}_{\mathrm{lag}6} = -0.87$ vs.\ $\mathrm{ac}_{\mathrm{lag}3} = 0.37$ on 11 accept-deltas (not supported---flipped; the accept-indexed observable is under-powered at 12 watermark points); vortex top-3 edit share $0.500$ over 6 edited docs (criterion ${\ge}0.5$: supported---upgraded from the pre-correction $41.7\%$ borderline).

\subsection{Giant scale probe (formal6, pre-correction instrument)}

$N{=}2000$, fan as-run $1$ (the driver header says 5; see the fan note in the supplementary README), $b{=}6$, 150 rounds, seed 11: $s_0 = 0.1245$ (train-based), final $0.240$, 4 accepts (increments $+0.057, +0.011, +0.025, +0.023$); laminar-window fit not administrable (insufficient post-transient points); wall 9281\,s. \emph{Caveat:} this probe ran under the pre-correction gate; its dilution reading is retained as pre-correction evidence, and the corrected-gate N-axis boundary remains to be re-measured. \emph{Superseded by formal11 (above):} under the corrected gate the same protocol yields 8 accepts and final $0.274$.

\subsection{Time-axis dense probe (formal7, pre-correction instrument)}

$N{=}50$, fan as-run $1$, $b{=}9$, 500 rounds, seed 11: final $0.4404$, 13 accepts; post-burst linear fit (skip 4, 10 points): $\hat\mu = 0.0171$/accept, $R^2 = 0.928$ (pre-registered bar $R^2{>}0.95$ not met); last-third slope $0.0087$/accept (drift alive); wall 16171\,s. Verdict: \emph{intermediate}---the entrance length exceeds 500 rounds at this density. \emph{Same pre-correction caveat as formal6. Superseded by formal12 (above):} under the corrected gate the run reaches $0.465$ with tail $0.003$/accept (NESS plateau).

Every headline number above is field-level traceable to the released run artifacts: each formal driver writes a result JSON under runs/ in the supplementary, and the field-level provenance table (paper number $\to$ source file/field $\to$ value) ships inside the supplementary package alongside them.

\section{Structural Correspondence: DSH $\leftrightarrow$ Hard-Sphere Assumptions}\label{app:corr}

\begin{tabularx}{\textwidth}{@{}X X X@{}}
\toprule
Hard-sphere/kinetic assumption & DSH plugin property & Cordis formalization \\
\midrule
permutation symmetry & order-independent composition & Independence, Defs.~17--19, Thm.~20 \\
microscopic reversibility / Liouville & mount--unmount reversibility & revertible effects, Thm.~61 \\
no instantaneous self-coupling & acyclic dependency graph & reactive coeffects, Thms.~63--64 \\
short-range contact interaction & typed contracts & commutative keys, Def.~39, Thms.~40--42 \\
\bottomrule
\end{tabularx}

A fifth result underpins the replacement primitive itself: the hot-replacement calculus is confluent (Thm.~73)---order-insensitive, residue-free skill substitution is a theorem, which is the formal foundation on which the bounded-edit discipline (H3$'$) and, historically, SkillOpt's empirical stability rest.

\begin{proposition}[Obstruction by construction]\label{prop:obstruct}
(i) If the trainable state is a single shared document (skill-library architecture), the empirical distribution is one Dirac mass: $Q[f,f] = 0$ identically (no distinct pairs), $N_{\mathrm{eff}} \equiv 1$, and every population-level statement of this paper is vacuous, not false. (ii) If the runtime is a star of singleton remote services (MCP architecture), A1 is untestable (no interchangeable version instances) and A2--A3 are obligations the runtime cannot discharge: side-effecting RPC carries no tracked inverse, and unrestricted call graphs admit instantaneous cycles. In both cases \eqref{eq:ktp} is undefined, not inaccurate.
\end{proposition}

\begin{proof}[Proof sketch]
(i) $Q$ is bilinear and its rate kernel counts ordered instance pairs; a singleton population has none, so all moments are ensemble-trivial. (ii) A1--A4 are interface obligations in the Cordis sense (checkable only where the runtime enforces structure); black-box remote semantics enforces none of them.
\end{proof}

\emph{Reading.} A1--A4 are obligations, not exclusions: a skill library extended with versioned instances, tracked inverses, enforced dependency acyclicity, and typed interfaces \emph{would} qualify---the proposition's content is that current everyday implementations do not provide these by construction, and that no kinetic statistics exist until they do.

\begin{table}[h]
\centering
\small
\caption{Why DSH-class runtimes, and not everyday skill/MCP stacks, support the gas description. Each axiom is a Cordis theorem for DSH; the everyday stacks fail at least one axiom, and both fail the zeroth requirement (a countable instance population with carried state).}
\label{tab:why-dsh}
\begin{tabular}{@{}p{3.1cm}p{3.35cm}p{2.95cm}p{3.15cm}@{}}
\toprule
Requirement & DSH plugin & Skill-document library & MCP tool stack \\
\midrule
instance population (zeroth) & registry of (plugin, version) instances with credit ledger & one shared document, edited in place & singleton remote services (star topology) \\
A1 permutation symmetry & holds, Thm.~20 & fails: prompt order/shadowing & partial: registration order irrelevant, but no interchangeable instances \\
A2 mount/unmount reversibility & holds, Thm.~61 (tracked inverses) & fails: in-place mutation, no inverse & fails: side effects, state kept server-side \\
A3 acyclic dependencies & holds, Thms.~63--64 & fails: free cross-referencing & fails: cyclic call graphs allowed \\
A4 typed contracts & holds, Def.~39, Thms.~40--42 & fails: natural-language interface & holds: JSON Schema I/O \\
\bottomrule
\end{tabular}
\end{table}

The correspondence is what licenses the gas vocabulary at all: each assumption has a concrete counterpart in the plugin runtime, and the ``kinetic lint'' procedure (per-plugin feature checklist, one rule per assumption) detects violations that would break the correspondence. The hard-sphere criterion reads ``hard'' as short-range, binary, single-timeslice-reversible---not elastic, which is the 2026-08-17 revision note: the driven-gas (inelastic) reading is what makes granular-gas clustering (Appendix~\ref{app:ops}, G.6) the operative physics instead of the classical elastic gas.

\section{Interpretation Boundaries}\label{app:labels}

Because the plugin gas does not conserve momentum ($\Sigma v$ changes under every operator; see the conservation ledger), no hydrodynamic limit in the Navier--Stokes sense exists, and we do not claim: (i) a turbulence cascade or any Kolmogorov-type spectrum; (ii) K\'arm\'an vortex streets (which require bluff-body wake instability from actual advection); (iii) Magnus effects (which require circulation and pressure differential in potential flow); (iv) Bernoulli's law (which requires pressure defined from momentum flux); (v) acoustic or gravity-wave dispersion. Where classical names appear in the experiment descriptions (``transition-like'', ``wave-like propagation''), they name an observable signature proposed by the analogy --- e.g., variance of acceptance intervals rising with retrieval radius; a reward depression propagating outward from a perturbed instance with damping --- and the claim is only that the signature was or was not observed. The theory's rigorous load-bearing elements are independent of the analogy: the hitting-time certificate (Theorem A, Lean-verified), monotonicity along the accepted lineage, the resolution law for gate power, and the fluctuation-scaling prediction (a CLT statement about population statistics).

The complete license inventory follows in Tables~\ref{tab:license-l0}--\ref{tab:license-boundary} give the full inventory of fluid-related notions used in this paper, classified by translation license, each with its device reading, its basis (architectural theorem, own-equation derivation, or required physics), its probe, and its measured verdict. The point of the inventory is the \emph{distribution} of verdicts: the license level organizes the outcome. Six of the eight L1 signatures are supported or observed (NESS relaxation included); tides flips to unsupported under the corrected gate on an under-powered observable; and the laminar window is bounded by the scale probes (plateau reached at 500 rounds, entrance persisting at $N{=}2000$). Of the L2 name-only candidates, Bernoulli, K\'arman, and waves fail as observables, and Magnus's bare correlation fires only after the gate correction ($r{=}0.04 \to 0.64$, $n{=}14$)---a gate-dependent transport reading that licenses no circulation claim. The mixed cases sit on the boundary: vortex now meets its pre-registered supported criterion, damping stays withdrawn, and the drive reading stays marginal. Since the triage labels are post-hoc (observables and thresholds were pre-registered, the categories were not), we report this as a descriptive classification of the verdict distribution, not a predictive test.

\begin{table}[h]
\centering
\small
\caption{\textbf{L0 structural correspondences} (licensed by the architecture, Appendix~\ref{app:corr}; not analogies).}
\label{tab:license-l0}
\begin{tabularx}{\textwidth}{@{}l X X X@{}}
\toprule
Kinetic notion & Device counterpart & Architectural basis & Status \\
\midrule
particle & plugin-version instance $(x,v,e)$ on the plugin$\times$version$\times$task grid & A1 permutation symmetry (Defs.~17--19, Thm.~20) & by construction \\
hard-sphere collision & co-activation in one pipeline $+$ credit/impulse exchange & short-range, binary, single-timeslice-reversible; typed contracts (Def.~39, Thms.~40--42) & instantiated: fan-3 vs fan-1 probe, final $0.50$ vs $0.30$ (formal9) \\
microscopic reversibility & mount/unmount reversibility & A2 revertible effects (Thm.~61) & by construction \\
no self-collision & acyclic dependency graph & A3 reactive coeffects (Thms.~63--64) & by construction \\
interaction range & retrieval radius (top-$k$, sharpness $\alpha$) & bounded-kernel condition of Thm.~B & measured (peak $\alpha{=}2$) \\
Maxwell daemon & validation gate, external to the gas & Thm.~E; gate score strictly monotone & verified, every run \\
energy injection & batch reward $R_{\mathrm{batch}}$ & driven (inelastic) gas clause & observed at every accept \\
dissipation & rejection drain (tie-rejection) & gate strictness & observed in gate logs \\
temperature & proposal variance under identical history (LLM bath) & thermostat identification & defined \\
Knudsen number $\Kn$ & per-round collision density (budget $b$, population $N$); $b{\to}\Kn$ mapping (heuristic) & regime clause & corrected dial: $\beta$ flat at $b{=}3$ $\to -0.58$ ($r^2{=}0.94$) at $b{=}9$ \\
NESS & drive--dissipation balance under the gate & conservation ledger & observed (150--500-round plateaus; formal12) \\
$H$-functional & population concentration; diversity $=-H$ & definition; no-gate drift form & (heuristic) \\
\bottomrule
\end{tabularx}
\end{table}

\begin{table}[h]
\centering
\small
\caption{\textbf{L1 own-equation signatures} (derived from the plugin master equation; predicted to hold). Verdicts from the pre-registered probes of \S\ref{sec:experiments}.}
\label{tab:license-l1}
\begin{tabularx}{\textwidth}{@{}>{\raggedright\arraybackslash}p{3.7cm} X l X@{}}
\toprule
Signature (derivation) & Device observable & Probe & Reading \& verdict \\
\midrule
fluctuation scaling (Prop.~D) & run-to-run dispersion vs.\ $N$ & formal8 (corrected gate) & flat $\hat\beta{=}-0.04$ at $b{=}3$ (undetectable, $J{=}2$); $\hat\beta{=}-0.58$, $r^2{=}0.94$ at $b{=}9$ --- \textbf{supported at density}, directional \\
intermittency, CV${\approx}1$ (Thm.~D $+$ discrete gate events) & acceptance-interval CV & formal13a F-Turb & CV $=0.39/0.78/0.87$; \textbf{supported} \\
boundary layer / cold start (driven-gas entrance region) & $s_0{\to}s_1$ jump then plateau & every run & \textbf{observed} \\
Coanda-like attachment (collision locality) & edits cluster near recently edited instances & formal13b & $41.2\%$ vs.\ $30\%$ baseline; \textbf{supported} \\
conduction-like flattening (retrieval as transport) & reward depression spreads and flattens with $\alpha$ & formal13a F-Conduction & $\alpha{=}4$ flattest (spread $0.031{\to}0.007$); \textbf{supported} \\
tides / period locking (driven NESS) & score autocorrelation at the round period & formal13b & corrected re-run: $ac_{\mathrm{lag}6}{=}-0.87$ vs.\ lag3 $0.37$ (11 accept-deltas, under-powered); \textbf{not supported} \\
Brownian drift under attribution (fluctuation $+$ drift) & watermark rises directionally with LLM attribution & formal13a F-Brownian & $0.223{\to}0.308$ monotone; \textbf{supported} \\
laminar developed flow (constant-drift regime) & additive waterline after the transient & formal11/12 & formal11: 8 accepts/150r at $N{=}2000$, entrance-bound ($r^2{=}0.88$); formal12: NESS plateau at 500 rounds (tail $0.003$/accept, $r^2{=}0.80$); \textbf{NESS observed}, additive window open \\
\bottomrule
\end{tabularx}
\end{table}

\begin{table}[h]
\centering
\small
\caption{\textbf{L2 name-only candidates} (require momentum conservation / a Navier--Stokes closure the gas does not have; predicted to fail as mechanisms---and none is claimed: three observables fail outright, Magnus's bare correlation proves gate-dependent).}
\label{tab:license-l2}
\begin{tabularx}{\textwidth}{@{}>{\raggedright\arraybackslash}p{3.7cm} X l X@{}}
\toprule
Classical name (required physics) & Why absent here & Probe & Reading \& verdict \\
\midrule
Bernoulli law (pressure from momentum flux) & $\Sigma v$ not conserved (ledger) & formal13b & fast-arm density undefined; \textbf{not supported} \\
K\'arman vortex street (bluff-body wake advection) & no advection term & formal13b & 0 alternating docs (pre-correction 1); \textbf{not supported} \\
Magnus effect (circulation $+$ pressure differential) & no potential flow & formal13b & bare correlation fires gate-dependently ($r{=}0.039$ pre-correction $\to 0.638$, $n{=}14$): a transport proxy, not circulation; \textbf{name not claimed} \\
gravity/acoustic waves (inertia $+$ restoring force) & no inertial term & formal13b & near/far drop $=0$; \textbf{not supported} \\
Kolmogorov cascade (NS inertial range) & no NS closure (\S\ref{sec:ness}) & not probed & \textbf{not claimed} \\
\bottomrule
\end{tabularx}
\end{table}

\begin{table}[h]
\centering
\small
\caption{\textbf{Boundary cases} (mixed L1/L2 character; the analogy's edge is genuinely fuzzy, and the readings are mixed --- itself a consistency check). Full per-probe narratives: Appendix~\ref{app:records}.}
\label{tab:license-boundary}
\begin{tabularx}{\textwidth}{@{}>{\raggedright\arraybackslash}p{3.7cm} X l X@{}}
\toprule
Phenomenon & Device observable & Probe & Reading \& verdict \\
\midrule
laminar--turbulent transition & acceptance-interval CV vs.\ control parameter $\alpha$ & formal13a F-Transition & non-monotone; \textbf{borderline} \\
vortex-like concentration & top-3 edit share & formal13b & $50.0\%$ over 6 edited docs; \textbf{supported}, small $n$ \\
viscous damping ($\mu$ is a naming convention) & perturbation decay at $b{=}0$ & formal13a F-Damping & \emph{withdrawn}: drive not actually off, confounded design \\
external drive without gate signal & forced $N{=}500$ arm & formal8/13a F-Drive & lift $+0.097$, below the $+0.10$ bar; \textbf{marginal} \\
\bottomrule
\end{tabularx}
\end{table}

\section{Operator Economy, Thermodynamics, and Regime Details}\label{app:ops}

\subsection*{F.1 Flow and external field}
During a forward pass the plugin set is frozen; occupied mass travels along the pipeline DAG layer by layer, and the skill coordinate, credit, and momentum are unchanged---pure transport, deterministic and bit-reproducible (the reproducibility/H2 property, not time-reversal symmetry; a replayed trajectory is not a microscopic crossing back of a wall). In this phase $Q=R=G=0$ and $\partial_t f + v\cdot\nabla_x f = 0$. The $v$-component of the flow along $x$ is the layer-depth progression; the action of the flow on the skill coordinate is currently degenerate (only edits move $s$), which is exactly the observation that the description's center of gravity sits with the collision operator today---if plugins ever carry continuous runtime state, flow takes over and the statistical description shifts from collisions to the flow field (heuristic).

The evaluation/workload category weights act as an external force $F_\tau(x)$: a load tilted toward one task type pulls the population toward that region of capability space. The LLM inference chain (provider hierarchy with failover) is the \emph{thermostat}: it injects noise and priors into both the collision rule (reflector) and the gate (verifier). The operationalization of temperature is the variance of edit proposals under identical history; a provider switch is a bath change, and a MCP failover is a transient temperature jump.

\subsection*{F.2 Collision structure, reaction channels, conservation ledger}

Two instances $i,j$ meet when they are co-activated in the same pipeline under task $\tau$, with rate kernel $K(x_i,x_j;\tau) \ge 0$; the collision cross-section is $\sigma_i = \int K(x_i,x';\tau)\,dx'\,d\tau$, and the contract width $d_i \propto \sqrt{\sigma_i}$ sets the amount of state exchanged per encounter. A collision event returns a batch reward $R_{\mathrm{batch}}$ from verification, and the reflector's attribution gives the post-state
\begin{equation}
e_i' + e_j' = e_i + e_j + R_{\mathrm{batch}},\qquad
v_i' = v_i + \eta\, a_i,\qquad
s_i' = s_i + \Delta s_i.
\end{equation}
Two structural points complete the picture. First, the collision is \emph{inelastic} in a driven-gas sense: $R_{\mathrm{batch}}$ is external work injected by task rewards, attribution redistributes it but total credit increases ($\Sigma e$ not conserved), and the edit impulse is likewise injection---``hard-sphere'' reads as short-range, binary, single-timeslice-reversible, not elastic. Second, the LLM is the \emph{sampler} of the collision rule: the explicit form of the post-state map is exactly what turns the collision kernel from a black box into an estimable object.

Reactions are species-level events aligned with the plugin lifecycle: append (nucleation, proposal rate from the exploration weight), delete (annihilation), update (lineage: a version increment, merged into collision), \emph{materialize} (gas $\to$ solid: sandbox trial mount becomes a registered plugin with a preset row), and unmaterialize (solid $\to$ gas). The key identification: sandbox trial-mounting is the \emph{gas phase} (free, reversible, load/unload at will), and materialization is \emph{solidification} (through the registry gate, entering the solid phase). The reaction-channel set is complete over capability space in the sense that any constitution-respecting upgrade is mountable---the separation of ``can'' (completeness of channels) from ``imagines'' (nucleation rate, the $p_{\min}$ of Theorem A) is the separation of chemical feasibility from nonzero reaction rate.

The conservation ledger records, quantity by quantity, what each operator does to the balance:
\begin{center}
\begin{tabularx}{\textwidth}{@{}>{\raggedright\arraybackslash}X >{\raggedright\arraybackslash}X >{\raggedright\arraybackslash}X >{\raggedright\arraybackslash}X >{\raggedright\arraybackslash}X@{}}
\toprule
Quantity & under Q & under R & under G & note \\
\midrule
particle number $N$ & conserved & not conserved (birth/death) & conserved & reactions change species count\\
total credit $\Sigma e$ & injected ($R_{\mathrm{batch}}$) & carried in/out & drained (rejected edits) & drive--dissipation balance $\Rightarrow$ NESS\\
total impulse $\Sigma v$ & not conserved & changes & projected (improving directions only) & directionality from $G$\\
gate score $V$ & --- & --- & monotone (strict) & verified Thm 5 monotonicity \\
\bottomrule
\end{tabularx}
\end{center}

\subsection*{F.3 Effective population gauges}
The fluctuation laws (Proposition~\ref{prop:fluct}) are statements about the \emph{effective} particle count $N_{\mathrm{eff}}$, which must be frozen before fitting (pre-registration discipline): (i) occupied-grid count = the number of (plugin $\times$ version $\times$ task) cells with nonzero credit in the run; (ii) total accepted edits; (iii) total credit $\Sigma e$ converted to an equivalent particle count. The gauges are measured, not assumed. In the active regime they track each other (formal1: $N{=}5$ arms give $3/1$ accepts and $\Delta V$ $+0.101/+0.011$; $N{=}50$ arms give $3/2$ accepts and $\Delta V$ $+0.189/+0.072$). At $N{=}500$ they diverge by orders of magnitude ($500$ occupied cells vs.\ $0/1$ accepts vs.\ $\Delta V$ $+0.000/+0.009$)---the pre-declared $10\times$-divergence warning fires where the rarefied reading places the device, i.e., the gauge check works as a designed diagnostic rather than holding by assumption. Correspondingly, the Proposition-D scaling is interpreted only inside the gauge-agreement regime; where the gauges diverge, the formula is diagnosed as failing, not confirmed. The rarefaction condition is the requirement that no instance dominates: $\max_i \sigma_i / \sum_j \sigma_j \ll 1$ (the bounded-kernel condition of Theorem B).

\subsection*{F.4 Theorem B: route choice and the enriched-state argument}
\emph{Honesty and route choice.} Not a new theorem: the standard result transposed, with an applicability verification (bounded kernel = no dominating particle; exchangeability). Two derivation routes exist for kinetic equations, and the choice matters. The \emph{Boltzmann--Grad route} (Lanford's short-time theorem~\cite{lanford1975}, extended to arbitrary long times for hard spheres by Deng--Hani--Ma~\cite{deng2024longtime}) starts from true deterministic binary collisions and must fight \emph{recollisions}---particles remembering whom they hit, which re-create correlations and break molecular chaos. The \emph{mean-field route} (Sznitman) assumes chaos and bounded interaction kernels. Our gas \emph{engineers recollisions by design}: the shared experience base makes particles re-encounter deliberately (retrieval is the re-collision channel), so a Boltzmann--Grad-type chaos-propagation claim is neither available nor desired, and the mean-field route is the honest level of description. The novelty is the operator set and what the gate does to the closure---with $G$ external, the system is not self-consistent in the classical sense; the daemon must remain outside the gas (Theorem E). The deeper reason the mean-field route is cleaner here than in physics: the hard-sphere obstacle is \emph{informational}---the one-particle density $f(x,v)$ does not record pair correlations, so recollision-built correlations progressively invalidate molecular chaos, and proving their negligibility took the 150-year program. The plugin gas loses nothing by construction: the instance state is \emph{enriched with its collision history}---the credit record $e$ summarizes whom the instance met and with what outcome, and the experience database is the population's stored collision history. Correlations that a gas loses are partially materialized as state, so we treat the one-instance description at $(x,v,e)$ as a \emph{Markov approximation}; residual dependence through the shared database is not proved negligible. Exchangeability is anchored operationally in A1 (permutation symmetry of composition, Cordis Thm.~20); the closure needs exchangeability plus a bounded kernel, not chaos surviving recollisions; retrieval locality provides the bound (an instance collides only within its tag neighborhood, so $\sigma_i$ is structurally finite). \emph{Measured premise check}: across the formal13a conduction arms (4 runs), the largest single-instance collision share is $\max_i \sigma_i / \sum_j \sigma_j = 0.067 \ll 1$---no dominating particle. This checks non-dominance over the measured runs---a diagnostic of the premise, not a proof of propagation of chaos.

\subsection*{F.5 Training thermodynamics}
With the Boltzmann convention $H[f] = \iiint f \log f\,dx\,dv\,de$, $H$ measures concentration (diversity reads as $-H$). The textbook statement---symmetric kernel and detailed balance imply $H$ monotone, relaxation to a Maxwellian---does \emph{not} hold here: the collision rule is sampled by an LLM and is asymmetric (proposals biased toward prior and retrieval history; attribution violates detailed balance). The correct no-gate form is a \emph{replicator--mutator} system: with credit injection and mutation but no selection, the population does not thermalize to the uniform Maxwellian but drifts to the \emph{prior modality} of the LLM, and diversity collapses monotonically (heuristic). Translated to learning: edit activity does not vanish (edits are copious)---but edits merely restate what the prior already knows, the learning signal is zero. This is the kinetic form of model collapse: self-referential generation without selection pressure contracts to the generator's prior; and it identifies the deepest theorem of the theory: \textbf{training's signal is supplied not by collisions but by the gate}; the strictness of the gate (strict inequality, tie rejection, frozen validation set) is the selection strength of evolutionary theory. It yields the tested distinction: \emph{edit frequency $\ne$ learning rate; learning rate = edit frequency $\times$ information content of the gate}.

With $G$ active, the system is driven and dissipative: collisions inject reward (drive), the gate drains rejected-edits' credit (dissipation). The drive--dissipation balance is not the no-gate prior modality but a NESS (non-equilibrium steady state): a peak around successful plugins with an exploration halo around it. Under this mapping: training convergence = relaxation to the NESS (the monotone accept sequence is the discrete projection of the NESS direction); the plateau = a NESS at vanishing entropy production (nothing left to adjudicate, drive fully cancelled by dissipation); and the gate acceptance rate $a(t)$ is the observational proxy for the entropy production rate. Escaping the plateau requires changing parameters, not applying force: raise temperature (edit budget / gate truncation), change the field (new task types), or anneal deliberately (avoid quenching to glass).

The Sagawa--Ueda-type bound~\cite{sagawa2010} $\langle W \rangle \le -\Delta F + kT\,I$ looks like a law the gate must satisfy---but a \emph{boundary statement is required}: defining free energy over capability space is not done, so within this system the bound is \emph{not} a quantitative thermodynamic conclusion; it is used as an information-theoretic analogy only. The quantitative channel is the Shannon one: the information content of the daemon's measurement channel, whose operational proxy is exactly the $\hat\delta_{\mathrm{gate}} - \hat\delta_{\mathrm{ind}}$ split of Proposition C. Its corollaries are the theory's sharpest readings: (i) a fixed validation set yields shrinking daemon efficiency---the population progressively adapts to the validation set, the daemon's measurements grow ``stale'' ($\delta_{\mathrm{gate}}>0,\ \delta_{\mathrm{ind}}\approx0$), which is \emph{the thermodynamic statement of Goodhart}, the violation of molecular chaos at the daemon--gas level (Prediction~4); (ii) fresh validation samples = information refill, restoring daemon efficiency; (iii) the constitution axiom (frozen/evolving separation) restated: if the daemon could itself evolve, monotonicity breaks [verified Theorem 7]---self-reference in thermodynamic language is the daemon becoming part of the gas.

\subsection*{F.6 Regime structure}
Before the regime map, we state the boundary of the analogy precisely, because it is the theory's load-bearing wall. The conservation ledger above shows total impulse $\Sigma v$ is \emph{not} conserved (edits inject impulse; the gate projects it). Without momentum conservation there is no Navier--Stokes closure, hence no literal turbulence cascade, no K\'arm\'an wake instability, no Magnus force, no Bernoulli law from momentum flux. We therefore claim \emph{none} of these. The fluid picture enters the paper in exactly one role: a \emph{systematic generator of falsifiable hypotheses}. Each classical phenomenon proposes an observable signature in the agent system (variance of acceptance intervals, locality of edit propagation, concentration of edit lineages); we measure the signature and report both outcomes. Where a signature matches, the finding is ``the analogy's prediction holds,'' a cross-domain observation; where it fails, the analogy's boundary is located. The transport-coefficient discussion below is a naming convention for measurable quantities (retrieval radius, edit damping, sampling noise), not a derivation of constitutive relations; the Euler-type system is displayed as a \emph{formal analogy} --- the moment hierarchy has the same shape --- and we do not perform, or claim, a Chapman--Enskog derivation for the plugin gas.

Control parameters: temperature $T$ (edit budget + LLM noise), density $\rho$ (plugins per unit capability volume), gate strictness. Three phases: \emph{gas} (free birth--death, no stable workflow motifs, jumpy curves, medium acceptance, membership volatile---early exploration); \emph{liquid} (co-activation clusters condense into recurring pipeline motifs, smooth increase, stable motifs with peripheral exploration---mid-stage convergence); \emph{solid} (frozen workflows, $T\to0$, flat curves, reproducible---released state). The microscopic mechanism of the phase picture sits with the reaction operator: sandbox mounting is the gas, materialization is the solidification step, and promotion along the lifecycle is the particle's path gas $\to$ liquid $\to$ solid. Granular-gas physics predicts that inelastically colliding driven granular systems must cluster: the uniform state is unstable, particles self-organize into clusters---the counterpart statement here is that populations with strong credit injection and weak selection drift toward motif condensation, and the gate is what keeps the system out of the collapsed state (heuristic). Where the liquid phase is deep, a nonzero-diffusion ``platform floor'' exists even at zero retrieval radius---retrieval is not free acceleration (heuristic).

In the continuum regime, moment closure gives the Euler-type system by construction
\begin{gather}
\partial_t\rho + \nabla\cdot(\rho u) = S_\rho,\qquad
\partial_t(\rho u) + \nabla\cdot(\rho u\otimes u + P) = \rho F + S_u,\notag\\
\partial_t(\rho E) + \nabla\cdot(\rho E u + Pu + q) = \rho F\cdot u + W_{\mathrm{drive}} - D_{\mathrm{gate}},
\end{gather}
with $E = \tfrac12|u|^2 + \tfrac{d_{\mathrm{eff}}}{2}T + \bar e$ (kinetic + temperature + credit internal energy), $W_{\mathrm{drive}} = \langle R_{\mathrm{batch}}\rangle\times$ collision rate, and $D_{\mathrm{gate}}$ the gate dissipation (with the Landauer~\cite{landauer1961} term as an analogy only, per the boundary statement). First-order (Navier--Stokes-type) corrections are written in a Chapman--Enskog-shaped formal analogy---viscosity $-2\mu\mathcal{S}$ and heat flux $q = -\kappa\nabla T$---with no expansion performed or claimed (analogy only). The {\em physical interpretation of the transport coefficients} is the theory's most direct object-level bridge: thermal conductivity $\kappa$ = KNN experience retrieval (successful mutations diffuse across task classes; $\kappa \approx$ retrieval radius${}^2\times$ retrieval frequency), viscosity $\mu$ = damping of edit perturbations in the population (disturbance--recovery experiments), diffusion $D$ = exploration noise (LLM sampling variance). The radius--specialization trade then reads as a transport trade: wider retrieval $\Rightarrow$ larger $\kappa$, faster equilibration of underperforming classes, but flatter specialization peaks---``retrieval breadth vs.\ specialization depth'' is not an engineering taste question but a heat-transport trade between diffusivity and peak gradient (semi-quantitative directional statement, labeled as such).

Finite $N$ couples conservation laws to white noise of strength $\propto\sqrt{T/N}$ (Landau--Lifshitz), giving the P1 prediction the pre-registered form: with the measurement statistics properly powered (arms spanning 2+ orders of magnitude in $N_{\mathrm{eff}}$; per-arm seed count fixed by a pre-registered power analysis, arms before seeds), run-to-run std should obey $\mathrm{std} \approx C\sqrt{T/N_{\mathrm{eff}}}$; a failure of decay signals excessive inter-particle correlation and the mean-field picture's breakdown---the diagnostic ordering between agent-level and machine-level explanations. This is quantified in the device measurements (Prediction~1). The regime clause is falsifiable from both sides: a powered probe at sparse density that nonetheless finds the $-1/2$ decay would falsify the gating, and a powered probe at the dense corner that still finds no decay would falsify the CLT reading itself.

\begin{table}[!tb]
\centering
\small
\caption{Corrected-gate verdict summary (L1 own-equation and boundary signatures; full license inventory: Tables~\ref{tab:license-l0}--\ref{tab:license-boundary}).}
\label{tab:verdicts}
\small
\begin{tabularx}{\textwidth}{@{}X X X X@{}}
\toprule
Signature & Verdict & Signature & Verdict \\
\midrule
Fluctuation scaling & undetectable ($b{=}3$) / $-0.58$ ($b{=}9$) & Intermittency (CV${\approx}1$) & supported ($0.39/0.78/0.87$) \\
Coanda attachment & supported ($41.2\%$) & Conduction flattening & supported ($\alpha{=}4$) \\
Tides period-locking & not supported (corrected re-run) & Brownian directional drift & supported ($0.223{\to}0.308$) \\
Boundary layer (cold start) & observed & Laminar developed flow & NESS plateau (500 r); entrance at $N{=}2000$ \\
L2: Bernoulli / Karman / waves & not supported & Magnus correlation & gate-dependent ($r$ $0.04{\to}0.64$); name not claimed \\
Vortex concentration & supported ($50\%$, corrected) & Viscous damping / drive-vs-gate & withdrawn / marginal ($+0.097$) \\
\bottomrule
\bottomrule
\end{tabularx}
\end{table}

\section{Deployed Practice Tools}\label{app:practice}

A practitioner takes home: the resolution law, the hitting-time horizon, reproducibility priced at $\mathrm{std}\sim\sqrt{T/N_{\mathrm{eff}}}$, the plateau dual-channel decision tree, and Theorem E as a design red line (never let the evaluator evolve with the evaluated). Four further deployables follow: reproducibility SLAs; silent-collapse alarms ($\delta_{\mathrm{gate}}$--$\delta_{\mathrm{ind}}$ split; fired on our device); gas--liquid--solid registry staging (edit concentration as monopoly metric); stopping-rule economics.

\end{document}